\documentclass[a4paper,superscriptaddress,11pt]{quantumarticle}
\pdfoutput=1

\usepackage[utf8]{inputenc}
\usepackage[english]{babel}
\usepackage[T1]{fontenc}
\usepackage{amsmath}
\usepackage{amsthm}
\usepackage{hyperref}
\usepackage{amssymb}
\usepackage{dsfont}
\usepackage{appendix}

\usepackage{tikz}
\usepackage{lipsum}

\definecolor{rev1}{rgb}{0,0,1}

\newtheorem{thm}{Theorem}
\newtheorem{deff}{Definition}

\newtheorem{lem}{Lemma}
\newtheorem{expl}{Example}
\newtheorem{oracle}{Oracle}

\newcommand{\poly}{\mbox{{\rm poly}}}

\newcommand{\dd}{\mbox{{\rm d}}}

\newcommand{\PP}{\rm P}
\newcommand{\FP}{\rm FP}
\newcommand{\NP}{\mbox{\rmfamily\textsc{NP}}} 
\newcommand{\coNP}{\rm coNP}
\newcommand{\FNP}{\mbox{\rmfamily\textsc{FNP}}} 
\newcommand{\TFNP}{\mbox{\rmfamily\textsc{TFNP}}} 

\newcommand{\BQP}{\rm BQP}
\newcommand{\FBQP}{\mbox{\rmfamily\textsc{FBQP}}}
\newcommand{\QMA}{\mbox{\rmfamily\textsc{QMA}}} 
\newcommand{\coQMA}{\mbox{\rmfamily\textsc{coQMA}}} 
\newcommand{\FQMA}{\mbox{\rmfamily\textsc{FQMA}}} 
\newcommand{\TFQMA}{\mbox{\rmfamily\textsc{TFQMA}}} 

\newcommand{\gapQMA}{\rm gapQMA}

\newcommand{\TFgapQMA}{\mbox{\rmfamily\textsc{TFgapQMA}}} 

\newcommand{\CQMA}{\rm QCMA}

\newcommand{\FCQMA}{\rm FQCMA}
\newcommand{\TFCQMA}{\mbox{TFQCMA}}

\newcommand{\MA}{\rm MA}

\newcommand{\FMA}{\rm FMA}
\newcommand{\TFMA}{\rm TFMA}

\newcommand{\BPP}{\rm BPP}
\newcommand{\FBPP}{\rm FBPP}

\newcommand{\QLLL}{\rm QLLL}

\newcommand{\HH}{\mathcal{H}}

\newcommand{\suppress}[1]{}

\newcommand{\SPECT}{\mbox{\rmfamily\rm{Spect}}}
\newcommand{\SPAN}{\mbox{\rmfamily\rm{Span}}}

\def\N{{\mathbb N}}

\begin{document}

\title{Total Functions in QMA}
\date{\today}

\author{Serge Massar}

\affiliation{Laboratoire d'Information Quantique CP224, Universit\'e libre de Bruxelles, B-1050 Brussels, Belgium.}
\email{smassar@ulb.ac.be}
\orcid{0000-0002-4381-2485}
\author{Miklos Santha}
\affiliation{CNRS, IRIF, Universit\'e Paris Diderot, 75205 Paris, France.}
\affiliation{Centre for Quantum Technologies \& MajuLab, National University of Singapore, Singapore.}


\maketitle

\begin{abstract}
The complexity class $\QMA$ is the quantum analog of the classical complexity class $\NP$. The functional analogs of $\NP$ and $\QMA$, called functional $\NP$ ($\FNP$) and functional $\QMA$ ($\FQMA$), consist in either outputting a (classical or quantum) witness, or outputting NO if there does not exist a witness. 
The classical complexity class Total Functional  $\NP$ ($\TFNP$) is the subset of $\FNP$ for which it can be shown that the NO outcome never occurs. $\TFNP$ includes many natural and important problems. Here we introduce the complexity class Total  Functional $\QMA$ ($\TFQMA$), the quantum analog of $\TFNP$. 
We show that $\FQMA$ and $\TFQMA$ can be defined in such a way that they do not depend on the values of the completeness and soundness probabilities. 
We provide examples of problems that lie in $\TFQMA$, coming from areas such as the complexity of $k$-local Hamiltonians and public key quantum money. In the context of black-box groups, we note that Group Non-Membership, which was known to belong to $\QMA$, in fact belongs to $\TFQMA$.  We also provide a simple oracle with respect to which we have a separation between $\FBQP$ and $\TFQMA$. 
\end{abstract}

\section{Introduction}

Classical complexity classes are generally defined as consisting of decision problems. But functional analogs of these classes can also be defined. The functional analog of $\NP$ is denoted $\FNP$ (Functional NP). As a simple example, the functional analog of the travelling salesman problem is the following: given a weighted graph and a length $\ell$, either output a circuit with length less than $\ell$, or output NO if such a circuit does not exist.
The functional analog of $\PP$, denoted $\FP$, is the subset of $\FNP$ for which the output can be computed in polynomial time.

Total functional $\NP$ ($\TFNP$), introduced in \cite{MP91} and which lies between $\FP$ and $\FNP$, is the subset of $\FNP$ for which it can be shown that the NO outcome never occurs. As an example,  factoring (given an integer $n$, output the prime factors of $n$) lies in $\TFNP$ since for all $n$ a (unique) set of prime factors exists, and it can be verified in polynomial time that the factorisation is correct.
 $\TFNP$ can also be defined as the functional analog of $\NP\cap\coNP$ \cite{MP91}. 

$\TFNP$ contains many natural and important problems, including factoring, local search problems\cite{JPY88,PSY90,Kr89}, computational versions of Brouwer's fixed point theorem\cite{P94} and finding Nash equilibria\cite{DGP09,CDT09}. Although there probably do not exist complete problems for TFNP,  there are many syntactically defined subclasses of $\TFNP$ that contain complete problems, and for which some of the above natural problems can be shown to be complete. For recent work in this direction, see \cite{GP17}.

The quantum analog of $\NP$ is $\QMA$ \cite{KSV02}. $\QMA$ has been extensively studied, and contains a rich set of complete problems, see e.g. \cite{B12}. These complete problems are all promise problems. For instance the most famous one, the $k$-local Hamiltonian problem, involves a promise that the ground state energy 
of the input $k$-local Hamiltonian is either less than $b$ or greater than $a$, with $a-b=1/q(n)$, for some polynomial $q(n)$, and the problem is to determine which is the case.

Functional $\QMA$, the problem of producing a quantum state that serves as witness for a $\QMA$ problem was first introduced in the unpublished manuscript \cite{JWB03}.
For instance the functional analog of the
$k$-local Hamiltonian problem is the following: given the classical description of  a $k$-local Hamiltonian, either output a state with energy less than $b$, or output NO if such a state does not exist.

In   \cite{JWB03} it was observed that there is no obvious reduction of $\FQMA$ problems to $\QMA$ problems. This should be opposed to the case of $\NP$ complete problems  for which finding a witness reduces to solving the decision problem.

It is well known that the definition of $\QMA$ does not depend on the values of the completeness and soundness probabilities, as they can be brought exponentially close to $1$ and $0$ respectively \cite{KSV02,MW05,NWZ09}.
We  discuss different definitions of Functional $\QMA$. 
We show that with an appropriate definition, based on the notion of eigenbasis of a quantum verification procedure, one can prove a similar amplification result.
 These theoretical considerations are the topic of Section \ref{PrelDef}.

In Section \ref{PrelDef} we also introduce the functional class $\TFQMA$ (Total Functional $\QMA$) as the subset of $\FQMA$ such that only the YES answer of the $\FQMA$ problem occurs, i.e. for all classical inputs $x$ there exists a witness.
Similarly to $\TFNP$, the 
problems in $\TFQMA$ are not promise problems, 
rather they have a structure such that one can prove that only the YES answer occurs.

The main aim of the present paper is to show that $\TFQMA$ is an interesting and rich complexity class. 
In Section \ref{Sec:ProbTFQMA} we provide  examples of problems that belong to $\TFQMA$. These  are related to problems previously studied in quantum complexity, such as commuting quantum $k$-SAT, commuting $k$-local Hamiltonian, the Quantum Lov\'asz Local Lemma (QLLL) \cite{AKS12} and public key quantum money based on knots \cite{FGHLS12}. We show how these problems can be adapted to fit into the $\TFQMA$ framework. Then in Section \ref{SEC:Rel} we consider relativized problems.  In the context of black-box groups, we show that Group Non-Membership, which was known to belong to $\QMA$ \cite{Watrous}, in fact belongs to $\TFQMA$. We also exhibit problems based on the Quantum Fourier Transform (QFT) and provide a simple oracle with respect to which there is a separation between $\FBQP$ and $\TFQMA$.

In the conclusion 
we present open questions raised by the present work.

\section{Definitions}\label{PrelDef}

\subsection{QMA}\label{SubSec:QMA}

We denote by $\HH_n$  the Hilbert space of $n$ qubits. 
For pure states we use the Dirac ket notation  $\vert \psi \rangle$, whereas for density matrices we just use the Greek letter $\rho$.
We denote by $I_n$ the identity matrix acting on $n$ qubits.

We denote by $\poly$ the set of all functions
$f : \N \to \N$, where $\N=\{1,2,...\}$, 
for which there exists a polynomial time
deterministic Turing machine that outputs $1^{f(n)}$ on input $1^n$. Note that if $f\in \poly$ then there exists a polynomial $q$ such that for all $n\in \N$, $f(n)< q(n)$.

Computational processes that can be carried out in polynomial time are sometimes called \emph{efficient}.

\begin{deff}{\bf Quantum Verification Procedure.} 
\label{def:qvp}
A \emph{quantum verification procedure} is a family of polynomial time uniform quantum circuits $Q=\{Q_{n} : n \in \N\}$  with $Q_n$ taking as input $(x, \vert \psi \rangle\otimes \vert 0^{k(n)} \rangle )$  where 
$x \in \{0,1\}^n$ is a binary string of length $n$, $\vert \psi \rangle$ is a state of  $m(n)$ qubits, and both $m= m(n)$ and $k= k(n)$ belong to $\poly$.
The last $k$ qubits, initialized to the state $\vert 0^k \rangle$, form the \emph{ancilla Hilbert space} $\HH_k$, 
and the $m$-qubit states $\vert \psi \rangle$ form the \emph{witness Hilbert space} $\HH_m$. The outcome of the run of $Q_n$ is a random bit which is obtained by measuring the first qubit in the computational basis.
We denote this outcome by $Q_{n}(x,\vert \psi \rangle)$, and we interpret the outcome $1$ as \emph{accept} and the outcome $0$ as \emph{reject}.
\end{deff}

Note that a quantum verification procedure can 
of course also take as input a mixed state $\rho$, rather than a pure state  $\vert \psi \rangle$. Mixed states can be written as convex combinations of pure states. The acceptance (rejection) probability for the mixed state is the convex combination of the acceptance (rejection) probabilities for the constituent pure states. Abusing slightly the notation, we use the same notation $Q_{n}(x,\rho)$ for the outcome of the quantum verification procedure on the mixed state $\rho$.

\begin{deff}{\bf (a,b)--Quantum Verification Procedure.} 
\label{def:qvpab}
Let $q\in \poly$, and let
$a,b : \N \rightarrow [0,1]$ be polynomial time computable functions
 which satisfy 
\begin{equation}
a(n)-b(n)\geq 1/ q(n)\ .
\end{equation}
We say that a quantum verification procedure $Q$ is an 
$(a,b)$-\emph{quantum verification  procedure} (or shortly an 
$(a,b)$-\emph{procedure})
if for every $x$ of length $n$,  one of the following holds:
\begin{eqnarray}
\exists \vert \psi \rangle\ :\  \Pr [Q_{n}(x,\vert \psi \rangle)=1]\geq a,\ \label{QMA1}\\
\forall \vert \psi \rangle\ :\  \Pr [Q_{n}(x,\vert \psi \rangle )=1]\leq b.\ \label{QMA2}
\end{eqnarray}
We call $a$ and $b$ the completeness and soundness probabilities of the quantum verification procedure.
\end{deff}

\begin{deff}
\label{Deff:QMA_and_coQMA }
{\bf QMA and coQMA.} Let $a,b$ be functions as in Definition$~\ref{def:qvpab}$.
The \emph{class} $\QMA(a,b)$ is the set of languages $L \subseteq\{0,1\}^*$ such that there exists an
$(a,b)$-procedure $Q$, where for every x, we have
$x \in L$ if and only if 
Equation~\eqref{QMA1} holds
(and consequently, $x \notin L$ if and only if 
Equation~\eqref{QMA2} holds).

We call $Q$ a 
\emph{quantum verification procedure for $L$}. For $x\in L$, we say that a $\vert \psi \rangle$ satisfying Equation~\eqref{QMA1}
is a \emph{witness} for $x$.

The \emph{class} $\coQMA(a,b)$ is the set of languages $L \subseteq \{0,1\}^*$ such that there exists an
$(a,b)$-quantum verification procedure $Q'$, where for every x, we have
$x \in L$ if and only if Equation~\eqref{QMA2} holds
(and consequently, $x \notin L$ if and only if 
Equation~\eqref{QMA1} holds).

\end{deff}

It is of course essential to understand to what extent the above definitions depend on the bounds $a$ and $b$. Obviously we can decrease $a$ and increase $b$: 
$\QMA(a,b) \subseteq \QMA(a',b')$
with $a'\leq a$ and $b'\geq b$, so long as  $a'(n)-b'(n) \geq 1/q'(n)$, for some  $q'\in \poly$.

But can one increase $a$ and decrease $b$?
This was first addressed by Kitaev who showed that the separation $a-b$ could be amplified to exponentially close to 1
by using multiple copies of the input state and multiple copies of the verification circuit \cite{KSV02}, that is
by increasing both $m$ and $k$. This was further improved in 
\cite{MW05} (see also \cite{NWZ09})  where it was shown that by running forwards and backwards the original quantum verification procedure, only one copy of the input state was needed to obtain the same amplification, 
that is one needs only increase $k$. 

\begin{thm}
{\bf QMA Amplification \cite{KSV02,MW05,NWZ09}.} 
\label{Thm-QMA-Amplification}
 For all  $a,b$ be functions as in Definition$~\ref{def:qvpab}$, for all $r\in \poly$,
we have
$\QMA(a,b) \subseteq \QMA(1-2^{-r},2^{-r})$. 
\end{thm}

As a consequence the precise values of the bounds $a$ and $b$ are irrelevant. Traditionally they are taken to be $2/3$ and $1/3$. 
We will do here the same.
\begin{deff}
\label{Deff:QMA}
We define the class $\QMA$ as $\QMA(2/3, 1/3).$
\end{deff}
We will come back to the $\QMA$ amplification procedure below.

We now turn to a particular kind of $(a,b)$-procedure which will be our main topic of study: 

\begin{deff}{\bf $a$-Total Quantum Verification Procedure.} 
\label{Deff:Totalqvp}
Let $a : \N \rightarrow [0,1]$ be a polynomially time computable function.
We say that a quantum verification procedure $Q$ is an 
$a$-\emph{total quantum verification  procedure} (or shortly an 
$a$-\emph{total procedure})
if for every $x$ of length $n$,  the following holds:
\begin{equation}
\exists \vert \psi \rangle\ :\   \Pr [Q_{n}(x,\vert \psi \rangle)=1]\geq a \ . \label{TQMA1}
\end{equation}
\end{deff}

Note that an $a$-total procedure is also an $(a,b)$-procedure for all $b$ satisfying the conditions of Definition \ref{def:qvpab}. 
 Note  that the language associated to an $a$-total procedure is $L=\{0,1\}^*$. That is the decision problem for total procedures is trivial, since for all $x \in \{0,1\}^*$ there exists a witness for $x$. Therefore for total procedures, the only interesting questions concern the witnesses.

 In order to prepare for a detailed study of total procedures, we therefore delve deeper into the structure of the witness space.

\subsection{Structure of the witness space.}
\label{SubSec:AmplStructureWitness}

 The methods used in \cite{MW05,NWZ09} to obtain Theorem \ref{Thm-QMA-Amplification} are based on Jordan's lemma \cite{J1875} (for a short proof of Jordan's lemma, see \cite{R06}). The use of Jordan's lemma in this context provides important insights into the structure of the witness space. A succinct proof of this structure was given in
\cite{ABBS08}. We state here these result which will play an important role in what follows.

\begin{thm}
{\bf Structure of witness space \cite{MW05,ABBS08}.} 
\label{Thm:BlockStructure}
Given a  
quantum verification procedure $Q=\{Q_{n}\}$, for all $x \in\{0,1\}^n$,
there exists a basis $B_Q(x)=\{\vert \psi_i\rangle : 1 \leq i \leq 2^m\}$ of the witness space  $\HH_m$
such that
the acceptance probability of linear combinations of the basis states does not involve interferences,
that is for all $\alpha_i$ such that $\sum_i \vert\alpha_i\vert^2=1$, we have
\begin{eqnarray}
 &\Pr [Q_{n}(x,\sum_i \alpha_i \vert \psi_i \rangle) = 1 ]&
\nonumber\\ 
& =
\sum_i \vert\alpha_i\vert^2  \Pr [Q_{n}(x,\vert \psi_i \rangle)
=1]&\ .
\label{Eq:eigenbasis}
\end{eqnarray}
\end{thm}

\begin{proof}
This result follows from the spectral decomposition of the POVM element corresponding to the quantum verification procedure giving outcome $1$, see beginning of Section 6.1 of  \cite{ABBS08} for details.
\end{proof}

\begin{deff}
\label{Deff:eigenbasis_spectrum_qvp}
{\bf Eigenbasis, spectrum and eigenspaces of a quantum verification procedure.}
Fix a  quantum verification procedure $Q=\{Q_{n}\}$, $x \in\{0,1\}^n$, and an eigenbasis $B_Q(x)=\{\vert \psi_i\rangle \}$ of $Q$ for $x$.

Given $\vert \psi_i\rangle\in B_Q(x)$, we call 
\begin{equation}
p_i= \Pr [Q_{n}(x,\vert \psi_i \rangle)
=1]
\end{equation}
 the \emph{acceptance probability} of $\vert \psi_i\rangle$.

We call the set of acceptance probabilities the \emph{spectrum} of $Q$  for $x$:
\begin{eqnarray}
\label{Eq:SPECT(Q,x)}
&\SPECT(Q,x) =\{p\in[0,1] :  \exists \vert \psi_i\rangle \in B_Q(x) &\nonumber\\
&\quad  {\rm such~that\ } 
\Pr [Q_{n}(x,\vert \psi_i \rangle) = 1 ]=p\}\ .&
\end{eqnarray}

Given  $p\in\SPECT(Q,x)$, we call 
\begin{eqnarray}
\label{Eq:HH(Q,x)}
\HH_Q(x,p)&=&\SPAN ( 
\{ \vert \psi_i \rangle \in B_Q(x)  \nonumber\\
& &\ :\  \Pr [Q_{n}(x,\vert \psi_i \rangle) = 1 ]=p\} )
\end{eqnarray}
the \emph{eigenspace} of $Q$ for $x$ with acceptance probability $p$.
\end{deff}

The eigenbasis is not necessarily unique: if two states $\vert \psi_i\rangle, \vert \psi_{i'}\rangle
\in B_Q(x)$ have the same acceptance probability, than  a unitary transformation acting on $\vert \psi_i\rangle, \vert \psi_{i'}\rangle$ yields a new eigenbasis. However, as the following result shows, this is the only freedom one has when choosing an eigenbasis.

\begin{thm}
{\bf Uniqueness of the spectrum and eigenspaces of $Q$.} 
\label{Thm:UniquenessBlockStructure}
Given a  
quantum verification procedure $Q=\{Q_{n}\}$ and $x \in\{0,1\}^*$,
the spectrum  $\SPECT(Q,x)$ of $Q$ 
and
the eigenspaces $\HH_Q(x,p)$ of $Q$ with acceptance probability $p\in \SPECT(Q,x)$ are unique and do not depend on the choice of eigenbasis $B_Q(x)$.
\end{thm}

\begin{proof}
Follows from the uniqueness of the spectral decomposition of the POVM element corresponding to the quantum verification procedure giving outcome $1$.
\end{proof}

\subsection{Relations}
\label{SubSec:Relations}

Consider a quantum verification procedure $Q$. In this section we are interested in the set of states
 on which $Q$ accepts with high probability. We are also interested in the set of states on which $Q$ rejects with high probability. This leads us to the following definitions.

\begin{deff}
\label{Deff:Accepting_and_rejecting_states}
{\bf Accepting and rejecting density matrices and subspaces.} 
Let $Q=\{Q_n\}$ be a quantum verification procedure and fix $a\in[0,1]$.

We define the following relations over binary strings and density matrices:
\begin{eqnarray}
R_{Q}^{\geq a}(x, \rho) = 1 &\text{ if }& \Pr [Q_{n}(x,\rho)=1]\geq a\ ,\nonumber\\
R_{Q}^{\leq a}(x, \rho) = 1 &\text{ if }& \Pr [Q_{n}(x,\rho)=1]\leq a\ .
\label{EQ:RQgeqleq}
\end{eqnarray}

Using the notion of eigenspace $\HH_Q (x, p)$ of $Q$ introduced previously, we define the following binary relations 
over binary strings and quantum states:
\begin{eqnarray}
\HH_{Q}^{\geq a}(x, \vert \psi \rangle) = 1 &\text{ if }& \vert \psi \rangle \in \SPAN (  \{ \HH_Q (x, p) \ :\  p\geq a\} ) \ ,
\nonumber\\
\HH_{Q}^{\leq a}(x, \vert \psi \rangle) = 1 &\text{ if }& \vert \psi \rangle \in \SPAN (  \{ \HH_Q (x, p) \ :\   p\leq a\} )\ .
\nonumber\\
\label{EQ:HHQgeqleq}
\end{eqnarray}
\end{deff}

To simplify notation, we  denote
\begin{eqnarray}
R_{Q}^{\geq a}(x)&=&
\left\{ \rho \ :\  R_{Q}^{\geq a}(x, \rho) = 1\right\}\ ,\\
R_{Q}^{\leq b}(x)&=&
\left\{ \rho  \ :\   R_{Q}^{\leq b}(x, \rho) = 1\right\}\ ,\\
\HH_{Q}^{\geq a}(x)&=&
\left\{ \vert \psi\rangle  \ :\    \HH_{Q}^{\geq a}(x, \vert \psi\rangle) =1\right\}\ ,\\
\HH_{Q}^{\leq b}(x)&=&\left\{ \vert \psi\rangle  \ :\    \HH_{Q}^{\leq b}(x, \vert \psi\rangle) =1\right\}\ ,
\end{eqnarray}
and we will generally express results in terms of the sets $R_{Q}^{\geq a}(x), R_{Q}^{\leq b}(x)$ and the subspaces $\HH_{Q}^{\geq a}(x), \HH_{Q}^{\leq b}(x)$, rather then the corresponding relations.

The following result explains how these definitions are related.

\begin{thm}{\bf Partial equivalence between accepting and rejecting density matrices and subspaces.}\label{Thm:RelationsFQMA}
 Let $a,b$ be functions as in Definition$~\ref{def:qvpab}$ and let $Q$ be an $(a,b)$-procedure. Then, 
\begin{enumerate}
\item 
we have the inclusion
\begin{eqnarray}
 \HH_{Q}^{\geq a}(x) &\subseteq&  R_{Q}^{\geq a}(x)\nonumber\\
  \HH_{Q}^{\leq a}(x) &\subseteq&  R_{Q}^{\leq a}(x)
\end{eqnarray}
 (where we view $\HH_{Q}^{\geq a}(x) $ and $  \HH_{Q}^{\leq a}(x) $  not as sets of pure states, but as the 
 sets of density matrices associated to these pure states);\\
 \item
and in the other direction, if 
 $R_{Q}^{\geq a}(x)$ is non empty, then
$ \HH_{Q}^{\geq a}(x)$ is non empty, while  if 
 $R_{Q}^{\leq a}(x)$ is non empty, then
$ \HH_{Q}^{\leq a}(x)$ is non empty
\end{enumerate}
\end{thm}

\begin{proof}
We consider the $\geq a$ case, the $\leq a$ case is similar.

Denote by $\{\vert \psi_i\rangle\}$ the basis of eiegnestates of $Q$ for $x$.

Point 1 is trivial: $ \HH_{Q}^{\geq a}(x)$ is constituted of all linear combinations of eigenstates $\vert \psi_i\rangle$  whose acceptance probability $p_i$ is greater or equal than $a$. Hence, using Eq. \eqref{Eq:eigenbasis}, these states belong to $ R_{Q}^{\geq a}(x)$.

Point 2 is also easy. Since $R_{Q}^{\geq a}(x)$ is non empty, there is at least one density matrix $\rho$ whose acceptance probability is greater or equal than $a$. We can write $\rho$ as a convex combination of pure states. At least one of these pure states must have acceptance probability greater or equal than $a$. 
We write this state in the basis of eigenstates as 
\begin{eqnarray}
\vert \psi\rangle &=& \sum_i \alpha_i \vert \psi_i \rangle\nonumber\\
&=&\sum_{i : p_i<a} \alpha_i \vert \psi_i \rangle
+\sum_{i : p_i\geq a} \alpha_i \vert \psi_i \rangle\ .
\end{eqnarray}
Equation  \eqref{Eq:eigenbasis} then implies that at least one of the terms in the sum over $i:p_i\geq a$ must be non vanishing.
\end{proof}

\subsection{Functional QMA}
\label{SubSec:FQMA}

Consider an $(a,b)$-procedure $Q$. We are interested in the functional task of outputting a witness for $Q$, and in defining the corresponding complexity class. 

At first sight, the definition should be in terms of the relation $R_{Q}^{\geq a}(x)$ as this characterises the set of density matrices that will accept with probability larger than the completness threshold $a$. Indeed, this approach  was followed in \cite{JWB03}. However using  $\HH_{Q}^{\geq a}(x)$ as basis for the definition of $\FQMA$ has advantages as we now discuss.

First,
$R_{Q}^{\geq a}(x)$ is not closed under linear combinations: if the projectors onto $\vert \psi \rangle$ and $\vert \psi' \rangle$ belong to $R_{Q}^{\geq a}(x)$, then the projector onto the linear combination 
$a \vert \psi \rangle+ b\vert \psi' \rangle$ does not necessarily belong to $R_{Q}^{\geq a}(x)$. On the other hand $\HH_{Q}^{\geq a}(x)$ is a subspace and hence closed under linear combinations.

Second,
$R_{Q}^{\geq a}(x)$ does not transform simply under the amplification procedure described in \cite{MW05,NWZ09}, while $\HH_{Q}^{\geq a}(x)$ does transform in a simple way, see Theorem \ref{Thm:FQMAamplifSR} below.

We therefore adopt a definition of $\FQMA$ based on the subspace $\HH_{Q}^{\geq a}(x)$.

For reasons that we discuss in the next paragraphs, we define Functional $\QMA$ in terms of two relations:
\begin{deff}
\label{Deff:FQMA(a,b)}
{\bf Functional QMA (FQMA).}
 Let $a,b$ be functions as in Definition$~\ref{def:qvpab}$.
The \emph{ class} $\FQMA(a,b)$ is the set 
$\{(\HH_Q^{\geq a}(x, \vert \psi \rangle), \HH_Q^{\leq b}(x, \vert \psi \rangle))\}$
of pairs of relations,
where $Q$ is an $(a,b)$-procedure. 
\end{deff}

In order to motivate the above definition, it is interesting to consider the following example.

\begin{expl}
\label{Example:No1}
Consider a function $\epsilon:\N \rightarrow [0,1/3)$ that decreases faster than $1/ {\poly}(n)$
for any polynomial $\poly(n)$, for instance $\epsilon(n)=2^{-n-2}$

The example consists of a quantum verification procedure $Q$ whose spectrum 
is the set
\begin{equation}
\SPECT(Q,x) = \{ \frac{1}{3}, \frac{2}{3}-\epsilon(n), \frac{2}{3}\}\ .
\end{equation}
\end{expl}

Consider the interaction between an all powerful prover and a verifier in $\BQP$. The prover wants to convince the verifier that he can produce a witness for $x$ for Example \ref{Example:No1}, which we view as 
a $(2/3,1/3)$-procedure. But it is impossible for the verifier (except possibly by using the structure of $Q$)  to differentiate with high probability between  the case where the prover sends him an eigenstate with acceptance probability $2/3$ (a valid witness), and 
an eigenstate with acceptance probability $2/3 - \epsilon(n)$ (not a valid witness).  On the other hand the verifier will reject with high probability if the prover provides an eigenstate with acceptance probability $1/3$. Thus from the point of view of the verifier it is important to characterise not only what are the valid witnesses, but also what are the states he will reject with high probability. Hence Definition \ref{Deff:FQMA(a,b)} involves two relations.

In the classical case of Functional $\NP$, there are only two kinds of certificates, those which are accepted with probability $1$ and those which are accepted with probability $0$, hence Functional $\NP$ can be described by a single relation. In the case of 
Functional $\QMA$, there are three kinds of states in the witness Hilbert space, those which are accepted with probability greater  than $a$, those which are accepted with probability less than $b$, and the states which have intermediate acceptance probabilities. Hence it is natural to describe Functional $\QMA$ by two relations.

\subsection{$\FQMA$ amplification}

We  now show how the QMA Amplification results \cite{MW05,NWZ09} apply to Functional QMA. We will show that, as for $\QMA$,  the bounds $a$ and $b$ that appear in the
the definition of $\FQMA(a,b)$ can be changed at will. The analysis  is based on  \cite{MW05,NWZ09}, but needs some new concepts, as we need to show that the structure of the witness space does not change  under amplification.

\begin{deff}{\bf Eigenspace preserving map of quantum verification procedures.}
\label{Deff:StrongReduction}
Let $Q$ and $Q'$ be two quantum verification procedures. We say that  there exists \emph{an eigenspace preserving map from $Q$   to $Q'$} if for all $x\in \{0,1\}^*$:  
\begin{enumerate}
\item
there exists a basis $B_Q(x)=\{\vert \psi_i\rangle \}$ of the witness Hilbert space $\HH_m$
which is a joint eigenbasis of $Q$ and $Q'$ for $x$;
\item there exists a polynomial time computable strictly 
increasing 
 function $f:[0,1]\to [0,1]$ 
 such that if
$p_i=\Pr [Q_{n}(x,\vert \psi_i \rangle)
=1]$ is the acceptance probability of $\vert \psi_i\rangle$ for $Q_n$, 
and $p'_i=\Pr [Q'_{n}(x,\vert \psi_i \rangle)
=1]$ is the acceptance probability of $\vert \psi_i\rangle$ for $Q'_n$,  
then $p'_i=f(p_i)$.
\end{enumerate}
In what follows we will refer to an eigenspace preserving map simply as an \emph{e--map}.
\end{deff}

As a consequence, if there exists an e-map from $Q$ to $Q'$, then  most questions about witnesses for $Q$ can be reduced to questions about witnesses for $Q'$, in particular $\HH^{\geq a}_Q (x)=
\HH^{\geq f(a)}_{Q'} (x)$ and  $\HH^{\leq b}_Q (x)=
\HH^{\leq f(b)}_{Q'} (x)$. However $R_{Q}^{\geq a}(x) \neq 
R_{Q'}^{\geq f(a)}(x)$ and $R_{Q}^{\leq a}(x) \neq 
R_{Q'}^{\leq f(a)}(x)$ which is one of the reasons why we define functional $\QMA$ in terms of $\HH^{\geq a}_Q (x)$ and $\HH^{\leq a}_Q (x)$.

The reason why we require that the function $f$ be polynomial time computable
is because we wish that the soundness and completeness thresholds  of $Q$ be mapped onto the soundness and completeness thresholds of $Q'$, where we recall
that the soundness and completeness thresholds must be polynomial time computable, see Definition \ref{def:qvpab}.
That is, if $Q$ is an $(a,b)$-procedure such that there exists an e--map from $Q$ to $Q'$
via $f$, 
then $Q'$ is an $a',b'$-procedure with $a'(n)=f(a(n))$ and $b'(n)= f(b(n))$. 

Note that eigenspace preserving maps are transitive: if there exists an e-map from $Q$ to $Q'$, and if there exists an e-map from $Q'$ to $Q''$, then there exists an e-map from $Q$ to $Q''$. Note also that if we require that the inverse $f^{-1}$ of the strictly increasing function $f$ in Definition \ref{Deff:StrongReduction} is also polynomial time computable, then eigenspace preserving maps are an equivalence relation.

\begin{thm}
{\bf $\QMA$ Amplification\cite{MW05}   is an eigenspace preserving map.} 
\label{Thm:FQMAamplifSR}
Let $Q$ be a quantum verification procedure.
Let 
$a,b$ be functions as in Definition$~\ref{def:qvpab}$.
For all $r\in \poly$ there  exists a quantum verification procedure $Q'$, such that there exists an e-map from $Q$ to $Q'$, and such that 
the polynomial time computable strictly increasing function $f$ that defines the e--map  (see Definition \ref{Deff:StrongReduction}) 
satisfies
$f(a)\geq 1-2^{-r}$ and $f(b)\leq 2^{-r}$.
\end{thm}

\begin{proof}
One checks that the $\QMA$ amplification procedure of \cite{MW05}   is an eigenspace preserving map with the above properties.
\end{proof}

Note that the amplification procedure of \cite{MW05} does not allow us to choose $f(a)$ and $f(b)$ arbitrarily.  For this reason we introduce the following deamplification procedure.

\begin{thm}
{\bf QMA Deamplification   is an eigenspace preserving map.} 
\label{Thm:FQMADeamplifSR}
Let $Q$ be a quantum verification procedure. Let 
$a,b$ and $a',b'$ be pairs of functions as in Definition$~\ref{def:qvpab}$
with $a\geq a' > b' \geq b$.
Then there exists a quantum verification procedure $Q'$, such that there exists an e-map from $Q$ to $Q'$, and such that 
the polynomial time computable strictly increasing function $f$ that defines the e--map  (see Definition \ref{Deff:StrongReduction}) 
satisfies
$f(a)=a'$ and $f(b)=b'$.
\end{thm}

\begin{proof}
We construct $Q'$ as follows. 

Let $z,z':\N \to [0,1]$, with $z>z'$, be two polyomial time computable functions to be fixed below.

On any input $(x,\vert \psi \rangle)$ run $Q$; if $Q$ accepts, then accept with probability $z\in[0,1]$ and reject with probability $1-z$; 
if $Q$ rejects, then accept with probability $z'\in[0,1]$ and reject with probability $1-z'$.

It is immediate to check that $Q$ e--maps to $Q'$, with the 
the strictly increasing 
function $f$ that defines the e--map (see Definition \ref{Deff:StrongReduction}) given by
\begin{equation}
f(p)=(z-z')p+z'\ .
\label{Eq:f(p)}
\end{equation}

We now solve for $z$ and $z'$ the equations $f(a)=a'$ and $f(b)=b'$. It is easy to check that $z,z'$ are rational functions of $a,b,a',b'$ hence polynomial time computable, that
 $z,z'\in [0,1]$, and that $z>z'$ since $a\geq a' > b' \geq b$.
\end{proof}

\begin{thm}
{\bf $\FQMA$ is independent of the bounds $(a,b)$.} 
\label{Thm:FQMASR}
Let $Q$ be a quantum verification procedure. Let 
$a,b$ and $a',b'$ be pairs of functions as in Definition$~\ref{def:qvpab}$
with $a' < 1-2^{-r}$ and $b'> 2^{-r}$,
for some $r\in \poly$.
Then there exists a quantum verification procedure $Q'$, such that there exists an e-map from $Q$ to $Q'$, and such that 
the strictly increasing function $f$ that defines the e--map  (see Definition \ref{Deff:StrongReduction}) 
satisfies
$f(a)=a'$ and $f(b)=b'$.
\end{thm}

\begin{proof}
We first use the amplification procedure of \cite{MW05}   to construct an intermediate quantum verification procedure $Q''$, such that there exists an e-map from $Q$ to $Q''$ (as follows from Theorem \ref{Thm:FQMAamplifSR}).
The parameters of the amplification procedure are chosen such that
the strictly increasing function $f_1$ that defines the e--map from $Q$ to $Q''$  (see Definition \ref{Deff:StrongReduction}) 
satisfies
$f_1(a)\geq 1-2^{-r}$ and $f_1(b)\leq 2^{-r}$.

We then apply to $Q''$ deamplification as in Theorem
 \ref{Thm:FQMADeamplifSR} to obtain the quantum verification procedure $Q'$. The parameters of the deamplification procedure are chosen such that
the strictly increasing function $f_2$ that defines the e--map from $Q''$ to $Q'$  (see Definition \ref{Deff:StrongReduction}) 
satisfies
$f_2(f_1(a))=a'$
and $f_2(f_1(b))=b'$.
\end{proof}

As a consequence of Theorem \ref{Thm:FQMASR}, the precise values of the bounds $a$ and $b$ are 
 irrelevant to the definition of $\FQMA$.
 Therefore, similarly to the definition of $\QMA$ we  make the following definition.
\begin{deff}
\label{Deff:FQMA}
We define the \emph{class} $\FQMA$ as $\FQMA(2/3, 1/3)$.
\end{deff}

\subsection{FBQP}
\label{SubSec:FBQP}

The class $\BQP$ is the set of decision problems that can be efficiently solved on a quantum computer.

\begin{deff}
\label{Deff:EffState}
{\bf Efficiently preparable states.} 
Let $m \in \poly$.
A family of density matrices $\{ \rho(x) : x\in\{0,1\}^n,  n \in \N\}$
is efficiently preparable if 
$\rho(x)$ acts on $\HH_{m(n)}$ and if
there exists 
a 
polynomial time 
uniform family of quantum circuits $Q=\{Q_{n} : n \in \N\}$ 
with
$Q_n$ taking as input $(x,  \vert 0^{k} \rangle )$
with 
 $x \in \{0,1\}^n$ 
and  $k\in \poly$ with $k \geq m$,
and where $\rho(x)$ is obtained by tracing out the last $k-m$ qubits of $Q_n(x)$.
\end{deff}

\begin{deff}
\label{Deff:Lpoly}
{\bf The language class $\BQP$.} 
\label{Def:BQP'}

$\BQP \subseteq \QMA$ is 
the set of languages  $L\subseteq \{0,1\}^*$ such that
there   exists:
\begin{enumerate}
\item  an $(2/3,1/3)$ quantum verification procedure $Q=\{Q_{n} : n \in \N\}$ with
$Q_n$ taking as input $(x, \vert \psi \rangle\otimes \vert 0^k \rangle )$, 
where 
 $x \in \{0,1\}^n$ is a binary string of length $n$,
$\vert \psi \rangle$ is a state of
$m$ qubits, with $m,k\in \poly$;
\item an efficiently preparable set of density matrices $\{\rho(x)\}$ where $\rho(x)$ acts on $\HH_m$;
\end{enumerate}
and where for every x, we have
 $x \in L$ if and only if 
\begin{eqnarray}
\Pr [Q_{n}(x,\rho(x))=1]\geq 2/3,\ \label{BQP3}
\end{eqnarray}
and
 $x \notin L$ if and only if 
\begin{eqnarray}
\forall \vert \psi \rangle, \Pr [Q_{n}(x,\vert \psi \rangle )=1]\leq 1/3.\ \label{QMA2B}
\end{eqnarray} 
 \end{deff}

 \begin{deff}{\bf Functional $\BQP$.}
\label{Deff:FBQP}
The class $\FBQP$ is the subset of pairs of relations 
$\{(\HH_Q^{\geq 2/3}(x, \vert \psi \rangle), \HH_Q^{\leq 1/3}(x, \vert \psi \rangle))\}$ in $\FQMA$, with $Q$ an $(2/3,1/3)$-procedure,
 such that
there exists 
an efficiently preparable set of density matrices $\{\rho(x)\}$ and for all $x$,
if $\HH_Q^{\geq 2/3}(x, \vert \psi \rangle)$ is non empty
then $\rho(x) \in R_Q^{\geq 2/3}(x, \vert \psi \rangle)$.
\end{deff}

\subsection{Total Functional QMA}
\label{SubSec:TFQMA}

We  now address  the central topic of this study, 
the subset of $\FQMA$ for which there always exists a witness. We had previously introduced $a$--total procedures in Definition
\ref{Deff:Totalqvp}. 
We can now define the corresponding functional classes.

\begin{deff}
\label{Deff:Total}
{\bf Totality.}
A pair of relations $(\HH^{\geq a}_Q(x, \vert \psi \rangle), \HH^{\leq b}_Q(x, \vert \psi \rangle))$ in $\FQMA(a,b)$ is called \emph{ total} if for all inputs $x$ there exists at least one witness
$\vert \psi \rangle $, i.e. if $\HH^{\geq a}_Q(x)$ is non empty.
\end{deff}

\begin{deff}
\label{Deff:TFQMA}
{\bf Total Functional  QMA (TFQMA).}
 Let $a,b$ be functions as in Definition$~\ref{def:qvpab}$.
The \emph{ class} $\TFQMA(a,b)$  is the set 
$(\HH_Q^{\geq a}(x, \vert \psi \rangle), \HH_Q^{\leq b}(x, \vert \psi \rangle))$
of pairs of total relations, i.e. the set of pairs of relations in $\FQMA$
where $Q$ is an $a$--total verification
procedure.

The \emph{ class} $\TFQMA=\TFQMA(1/3,2/3)$ is the set of total relations in $\FQMA$.
\end{deff}

We emphasize that problems in $\TFQMA$ are not promise problems:
they satisfy that for all $x$ there exists at least one witness.
In analogy with $\FNP$ and $\TFNP$, we expect  problems in $\TFQMA$ to be simpler than general problems in $\FQMA$.

\subsection{Gapped quantum verification procedures}
\label{SubSec:GapQVP}

A sub-class of quantum verification procedures which will be of interest are those which have a gap in their spectrum. They are defined as follows.

\begin{deff}
{\bf Gapped (a,b)--Quantum Verification Procedure.} 
\label{def:qvpabGAP}
An  $(a,b)$-{procedure} $Q$ is a \emph{ gapped} $(a,b)$-\emph{ 
procedure} if for every $x$ of length $n$, there are strictly no eigenstates with acceptance probability 
comprised between $a$ and $b$. As a consequence the spaces $\HH^{\geq a}_Q (x)$ and $ \HH^{\leq b}_Q (x)$ generate the entire witness Hilbert space:
\begin{equation}
\HH_m =\SPAN (\HH^{\geq a}_Q (x) \cup \HH^{\leq b}_Q (x))\ .
\end{equation}
\end{deff}

\begin{deff}
\label{Deff:gapQMA}
{\bf gapQMA.} Let $a,b$ be functions as in Definition$~\ref{def:qvpab}$.
The \emph{ class} $\gapQMA(a,b)$ is the set of languages $L \subseteq\{0,1\}^*$ such that there exists a gapped
$(a,b)$-procedure $Q$, where for every x, we have
$x \in L$ if and only if 
Equation~\eqref{QMA1} holds
(and consequently, $x \notin L$ if and only if 
Equation~\eqref{QMA2} holds).
\end{deff}

\begin{deff}
\label{Deff:TFgapQMA}
{\bf Total Functional  gap QMA (TFgapQMA).}
 Let $a,b$ be functions as in Definition$~\ref{def:qvpab}$.
The \emph{ class} $\TFgapQMA(a,b)$ is the set of pairs of total relations 
$\{(\HH^{\geq a}_Q(x,\vert \psi\rangle), \HH^{\leq b}_Q(x,\vert \psi\rangle)\}$ for which 
$Q$ is a gapped $(a,b)$-quantum verification procedure.

We define the class $\TFgapQMA$ as $\TFgapQMA(2/3,1/3)$.
\end{deff}

\subsection{$1$-- and/or $0$--Quantum Verification Procedures }
\label{SubSec:10QVP}

Because of the $\FQMA$ amplification theorem, in all the above definitions we can  replace the upper bound $a$ (by convention taken to be $2/3$) in the definition by
$1-2^{-r}$, and the lower bound $b$ (by convention taken to be $1/3$) by $2^{-r}$, for any polynomial $r$. 
However sometimes one can show that one can take $a=1$ and/or $b=0$, which is potentially a stronger statement. 
In particular  the inclusions 
$\TFgapQMA(1,0) \subseteq \TFgapQMA \subseteq \TFQMA$
might be strict.

When $a=1$ and/or $b=0$ we are dealing with  \emph{exact quantum computation}. This means that the quantum circuit is made out of a finite (or possibly enumerable) set of quantum gates, and all operations (state preparation, gates, measurements in the computational basis) are  implemented with zero error.
Exact quantum computation has been studied in several contexts. In particular $\QMA_1
$ is the subset of QMA in which the accepting probability in the case of YES instances is 1, i.e.
$\QMA_1=\QMA(1,1/3)$. Complete problems for $\QMA_1$ were described in \cite{B11,GN16}. For arguments why it appears difficult to prove that $\QMA=\QMA_1$, see \cite{A09B}.

To simplify notation when dealing with exact quantum computation, we write
\begin{eqnarray}
R_{Q}^{1}(x)&= &R_{Q}^{\geq1}(x)\ ,\\
R_{Q}^{0}(x)& = &R_{Q}^{\leq 0}(x)\ ,\\
\HH_{Q}^{1}(x)&=&\HH_{Q}^{\geq 1}(x)\ ,\\
\HH_{Q}^{0}(x)&=&\HH_{Q}^{\leq 0}(x)\ .
\end{eqnarray}

\section{Problems in TFQMA}\label{Sec:ProbTFQMA}

\subsection{Preliminary Considerations}

\subsubsection{Introduction}

In this section and the next one we provide examples of problems in $\TFQMA$ which are not obviously in $\FBQP$.

At the heart of any such example is a quantum verification procedure $Q$. The minimum requirements for an example to be included in our list is to be able to show that $Q$ is an $a$-total quantum verification procedure, i.e. that $R_{Q}^{\geq a}(x)$ is non empty for all $x$ (and consequently $\HH_{Q}^{\geq a}(x)$ is non empty).
However in some cases we can also determine the eigenbasis $B_Q(x)=\{\vert \psi_i\rangle \}$ of $Q$ for $x$, and (at least partially) characterise the pair of relations $(\HH^{\geq a}_Q(x, \vert \psi \rangle), \HH^{\leq b}_Q(x, \vert \psi \rangle))$ in $\TFQMA(a,b)$.
The precise formulation of the examples will depend on how much we can say about them.

When possible we shall formulate the problems in such a way that they belong to $\TFgapQMA(1,0)$, as this is the strongest statement.

For convenience in what follows, the input size will not necessarily be denoted by $n$, which can be reserved
for some parameter in the input.

\subsubsection{Measurements}

It is often convenient to describe  some of the steps of a quantum verification procedures as a measurement. By this we mean that we carry out  an \emph{ ideal} quantum measurement (also known as a \emph{ quantum non demolition measurement}) which is realised as  a unitary transformation. Consider the operator $A=\sum_a a \Pi_a$, where the $a$'s 
are the eigenvalues of $A$ and $\Pi_a$ are orthogonal projectors. Measuring $A$  on state $\vert \psi \rangle$  corresponds to carrying out
 the unitary evolution
\begin{equation}
\vert \psi \rangle \vert 0\rangle
\to \sum_a  \Pi_a \vert \psi \rangle \vert a\rangle
\end{equation}
 where the second register is an ancilla that registers the outcome of the measurement. Subsequent operations can then be carried out conditional on the state of the ancilla.
 
 A specific measurement we will use several times is the projection on the antisymmetric and symmetric spaces. This can be realised by implementing the SWAP test \cite{BCWD01} as follows:
one carries out a Hadamard transform on an ancilla, then a conditional SWAP, and finally a Hadamard on the ancilla.
 \begin{eqnarray}
  \vert \phi \rangle \vert \psi \rangle \vert 0\rangle
  &\to & \frac{1}{\sqrt{2}}
    \vert \phi \rangle \vert \psi \rangle\left ( \vert 0\rangle +  \vert 1\rangle\right)\nonumber\\
    &\to &
    \frac{1}{\sqrt{2}}
   \left (  \vert \phi \rangle \vert \psi \rangle \vert 0\rangle + \vert \psi \rangle \vert \phi \rangle  \vert 1\rangle\right)\nonumber\\
    &\to &    
       \frac{1}{2}
   \left (  \vert \phi \rangle \vert \psi \rangle  + \vert \psi \rangle \vert \phi \rangle \right) \vert 0\rangle\nonumber\\
   & &
+
     \frac{1}{2}
   \left (  \vert \phi \rangle \vert \psi \rangle  - \vert \psi \rangle \vert \phi \rangle \right) \vert 1\rangle 
\end{eqnarray}
After these operations, if the ancilla is in the $\vert 0\rangle$ state one has projected onto the symmetric space, while if it is in the $\vert 1\rangle$ state one has projected onto the anti--symmetric space.
In the first case we say that the SWAP test has outputted "Symmetric", while in the second case that it has outputted "Anti-symmetric".

\subsection{Eigenstates of commuting $k$-local Hamiltonian}

\subsubsection{Background}

\begin{deff}
\label{Deff:klocalH}
{\bf $k$-local Hamiltonian.} 
Fix $k,d\in \N$. 
A qudit is a quantum system of dimension $d$. Denote by $\HH$ the Hilbert space of $n$ qudits.
Let $A\in \poly$.
A $k$-local Hamiltonian is a Hermitian matrix acting on $\HH$ which can be written as
$H=\sum_{a=1}^{A(n)} H_a$,
where each term $H_a$ (sometimes called constraint) is a Hermitian operator that acts non trivially on at most $k$ qudits and whose matrix elements in the computational basis have an efficient classical description.
\end{deff}

The $k$-local Hamiltonian problem is to determine whether the ground state of $H$ has energy $\leq b$ or $\geq a$, with $a-b\geq 1/\poly(n)$, for some polynomial $\poly(n)$, with the promise that only 
one of these cases occurs. 
The $k$-local Hamiltonian problem is QMA complete~\cite{KSV02,KR03} even when $k=2$~\cite{KKR06}.

The \emph{ commuting} $k$-local Hamiltonian is the case where the operators $H_a$ commute. 
It was shown by Bravyi and Vyalyi that the commuting 2-local Hamiltonian problem is in NP~\cite{BV}. Some additional cases of commuting $k$-local Hamiltonian problems also in NP are: the 3-local Hamiltonian where the systems are qubits~\cite{AE11}, the 3-local Hamiltonian where the systems are qutrits and the interaction graph is planar or more generally nearly Euclidean~\cite{AE11}, 
the planar square lattice of qubits with plaquette-wise interactions~\cite{S11}; approximating the ground state energy when the interaction graph  is a locally expanding graph~\cite{AE15}. The complexity of the commuting $k$-local Hamiltonian problem in the general case is unknown.

A particularly interesting case is when each $H_a$ is a projector, that is has only $0,1$ eigenvalues. The $k$-local Hamiltonian problem in this case reduces to the question whether $H$ has a frustration free eigenstate, an eigenstate with eigenvalue $0$. This is known as quantum $k$-SAT ( denoted $k$-QSAT), and was introduced in~\cite{B11} where it was shown that 2-QSAT is in P and  $k$-QSAT for $k\geq 4$ is 
$\QMA_1$ complete (where $\QMA_1$ is the subset of QMA in which the accepting probability in the case of YES instances is 1). It was later shown that  3-QSAT is also $\QMA_1$ complete~\cite{GN16}.

\subsubsection{Notation}

Consider a commuting $k$-local Hamiltonian $H=\sum_{a=1}^A H_a$, $[H_a,H_{a'}]=0$, acting on $n$ qubits.
Since the $H_a$'s are hermitian and commute, they possess a common eigenbasis. That is, there exists a basis ${\cal B} = \{\vert \psi_{hj}\rangle\}$ of the Hilbert space 
$\HH_n$, where each basis state $\vert \psi_{hj}\rangle$ is also an eigenstate of all the constraints $H_a$:
\begin{eqnarray}
H_a \vert \psi_{hj}\rangle &=& h_a \vert \psi_{hj}\rangle\ ,\\
\langle \psi_{h'j'}\vert \psi_{hj}\rangle &=&\delta_{h'h}\delta_{j'j} \ .
\end{eqnarray}

We denote by
$h=(h_1,...,h_A)\in \mathbb{R}^A$    the string of eigenvalues of $H_a$. The same symbol $h$ is also used as the  first index  in labelling the basis states $\vert \psi_{hj}\rangle$.
The second index $j\in J_h$ labels orthogonal states within the subspaces with  the same eigenvalue $h$.

We denote by $E$ the energy of the eigenstate:
\begin{equation}
E=\sum_{a=1}^A h_a\ .
\end{equation}
We denote
by 
\begin{equation}
{F}=\{ (h,j)\}
\end{equation}
 the sets of indices of the basis ${\cal B}$. Since ${\cal B}$ is a basis we have $\vert {F} \vert = 2^n$. 
We denote by 
\begin{equation}
{G}=\{ h: \exists j ~(h,j) \in F\}
\end{equation}
the set of possible  strings of eigenvalues of $H_a$. 
Since there may exist orthogonal eigenstates with the same string of eigenvalues,  $\vert G \vert \leq 2^n$ with equality not necessarily attained.

Note that 
given  an eigenstate  $\vert \psi_{hj}\rangle$, one can efficiently determine the string $h=(h_1,...,h_A)$ 
by measuring each $H_a$ in succession, where the order is immaterial since the $H_a$ commute.

In the case where each $H_a$ is a projector,  $h_a \in \{0,1\}$, and $h \in \{0,1\}^A$.

\subsubsection{Frustration-Free or Degenerate Eigenspace of commuting quantum k-SAT}
\label{subsubsectionqkSAT}

\begin{deff} \label{Prob:QkSAT}
{\bf Frustration free or degenerate eigenspace of commuting quantum $k$-SAT with $n$ constraints.} 
Denote by $x$ the classical description of a
 commuting $k$-local Hamiltonian acting on the space of $n$ qubits $\HH_n$,  with $A = n$  constraints with $0,1$ eigenvalues (projectors), and where by hypothesis 
  each constraint can be measured with zero error in polynomial time using a quantum computer.
 
Denote by $\HH^1(x)$  the subspace of the space of $2n+1$ qubits $\HH_{2n+1}$ spanned by 
states satisfying one of the two conditions:
\begin{enumerate}
\item States with the first qubit set to $0$, the next $n$ qubits a frustration free state, i.e. a state such that 
 its eigenvalue sequence is $h=0^A$, and the last $n$ qubits are in an arbitrary state;
 \item States with the first qubit set to $1$ and the remaining $2n$ qubits  the antisymmetric linear combination of  two orthogonal eigenstates  with the same eigenvalue sequence $h=(h_1,...,h_n)$:
 \end{enumerate}
\begin{eqnarray}
\label{EQ:HH1X}
\HH^1(x)&=&\SPAN
( \nonumber\\
& &
\{
\vert 0\rangle \vert \psi_{ 0^Aj}\rangle \vert \psi_{h'j'}\rangle \ :\  j\in J_{0^A},  (h',j')\in F
\}
\nonumber\\
& \cup &
 \{
\frac{1}{\sqrt{2}}\vert 1\rangle\left( \vert \psi_{hj} \rangle  \vert \psi_{hj'} \rangle - 
\vert \psi_{hj'}\rangle \vert \psi_{hj} \rangle\right)\nonumber\\
& & \quad \quad
\ :\ h\in G, j,j' \in J_h , j\neq j' \}
\nonumber\\
& ).&
\end{eqnarray} 
 
  Denote by $\HH^0(x)$ the orthogonal subspace:
  \begin{eqnarray}
\label{EQ:HH0X}
&\HH^0(x)&=\SPAN (
 \nonumber\\
& &
\{ \vert 0\rangle \vert \psi_{ hj}\rangle \vert \psi_{h'j'}\rangle \ :\  (h,j), (h',j')\in F , h\neq 0^A\} \nonumber\\
& & \cup \{
\vert 1\rangle\vert \psi_{hj} \rangle  \vert \psi_{h'j'} \rangle
\ :\  (h,j), (h',j')\in F , h\neq h'  \} \nonumber\\
& &\cup \{
\vert 1\rangle  \vert \psi_{hj} \rangle  \vert \psi_{hj} \rangle 
\ :\ (h,j)\in F \}
\nonumber\\
& &\cup \{
\frac{1}{\sqrt{2}}\vert 1\rangle \left( \vert \psi_{hj} \rangle  \vert \psi_{hj'} \rangle + 
\vert \psi_{hj'}\rangle \vert \psi_{hj} \rangle\right)\nonumber\\
& & \quad \quad
\ :\  h\in G, j,j' \in J_h, j\neq j'  \}
\nonumber\\
& & )
\nonumber\\
\end{eqnarray}
\end{deff}

\begin{thm} The pair of subspaces $(\HH^1(x),\HH^0(x))$ defined in Definition~\ref{Prob:QkSAT} belong to $(1,0)$-total functional gap $\QMA$: 
\begin{equation}
(\HH^1(x),\HH^0(x)) \in \TFgapQMA (1,0)\ .
\end{equation}
\label{Thm:QkSAT}
\end{thm}

\begin{proof}
We will show 
that there is a quantum verification procedure $Q$ that accepts with probability $1$ on the states in $\HH^1(x)$ and accepts with probability $0$ on the states in $\HH^0(x)$. Furthermore we will show that $\HH^1(x)$ is non empty for all $x$. This then implies that $Q$ is a $(1,0)$ gapped total quantum verification procedure, and $(\HH^1(x),\HH^0(x)) \in \TFgapQMA (1,0)$.

{\bf Quantum verification procedure.}

We first describe $Q$.

Measure the first qubit, obtaining outcome $b$.

If $b=0$, 
measure all the $H_a$ on the first $n$ qubits, let the result be $h$. Accept if
 $h=0^A$, and otherwise reject. 
 
If $b=1$, measure all the $H_a$ both on the first $n$ qubits and on the second $n$ qubits, yielding two $n$ bit 
eigenvalue sequences
$h=(h_1,...,h_n)$ and $h'=(h'_1,...,h'_n)$. Reject if $h\neq h'$. 

If $h=h'$ carry out a SWAP test. Reject if the SWAP test outputs "Symmetric", accept if the SWAP test outputs 
"Antisymmetric".

{\bf $Q$ is a gapped $(1,0)$-procedure.}

Note that the states enumerated in Equations~\eqref{EQ:HH1X} and \eqref{EQ:HH0X} form a basis of the space of $2n+1$ qubits. Consequently we have $\HH_{2n+1}=\SPAN(\HH^1(x) \cup \HH^0(x))$.

It is straightforward to check that $Q$ leaves all the states enumerated in Equations~\eqref{EQ:HH1X} and \eqref{EQ:HH0X} invariant, and accepts with probability $1$ on the states in  Equation~\eqref{EQ:HH1X}, and rejects with probability $1$ on the states in  Equation~\eqref{EQ:HH0X}. Therefore $Q$ is a gapped $(1,0)$-procedure, and the states enumerated in Equations~\eqref{EQ:HH1X} and \eqref{EQ:HH0X} form an eigenbasis of $Q$ for $x$.

{\bf $Q$ a total procedure.}

To prove that $Q$ is total, we show that $\HH^1(x)$ is non empty for all $x$.

To this end, for every basis vector $\vert \psi_{hj} \rangle \in {\cal B}$  we consider its associated eigenvalue sequence 
$h=(h_1,...,h_n) \in G$. 
The basis $\cal B$ comprises $2^n$ states. The number of associated eigenvalue sequences $\vert G \vert$ is less or equal than $2^n$. 
Therefore, by the pigeonhole principle, either $\vert G \vert =2^n$ and there is a one to one mapping between basis states and bits strings, 
in which case  there is a basis state $ \vert \psi_{hj} \rangle$ with 
eigenvalue sequence $h=0^A$; or $\vert G \vert < 2^n$ and there is 
 at least one collision, i.e. at least two basis states with the same eigenvalue sequence.

In the first case  a  witness is provided by the frustration free state 
  \begin{equation}
  \vert 0\rangle \vert \psi_{0^Aj} \rangle \vert \psi'\rangle\ ,
  \label{Eq:QkSAT1}
\end{equation}
where $\vert \psi'\rangle$ is any state of $n$ qubits.

 In the second case there exists a state of the form
 \begin{equation} 
\frac{\vert 1\rangle}{\sqrt{2}}\left( \vert \psi_{hj} \rangle  \vert \psi_{hj'} \rangle - 
\vert \psi_{hj'} \rangle \vert \psi_{hj} \rangle\right),
\label{Eq:37}
\end{equation}
where $\vert \psi_{hj}  \rangle$ and $\vert \psi_{hj'}  \rangle$ are different basis vectors with the same eigenvalue sequence.
This state is also accepted by probability 1 by $Q$.

Hence $\HH^1(x)$ is non empty for all $x$.
\end{proof}

  Note that if one or more of the constraints $H_a=I_n$ is the identity operator, then there is no frustration free state, and the witnesses are necessarily of the form given in  Equation~\eqref{Eq:37}.

Note that the existence argument in the proof of Theorem \ref{Thm:QkSAT} 
is based on the pigeonhole principle, and therefore the problem \emph{frustration free or degenerate eigenspace of commuting quantum $k$-SAT with $n$ constraints}
has a form very similar to the problems
in the  Polynomial Pigeonhole Principle (PPP) class introduced in \cite{P94}. 
It has the following classical analog: given a $k$-SAT formula with $n$ variables and $n$ clauses, either find a satisfying assignment, or find two assignments such that the clauses all have the same value.

\subsubsection{Almost degenerate states of commuting $k$- Hamiltonian.}

\begin{deff} 
\label{Prob:klocalH}
{\bf Almost degenerate eigenspace of commuting $k$-local Hamiltonian.}
Denote by $x$ the classical description of a
 commuting $k$-local Hamiltonian acting on the space of $n$ qubits $\HH_n$, $H=\sum_{a=1}^A H_a$, with the local terms   bounded by $0\leq H_a \leq I_n/A $,
 and where by hypothesis 
  each $k$-local term can be measured with zero error in polynomial time using a quantum computer, and where by hypothesis the eigenvalues of the $k$-local terms can be efficiently computed, and efficiently added and subtracted with zero error.
 
 Denote by $\HH^1(x)$  the subspace of the space of $2n$ qubits $\HH_{2n}$ spanned by 
the antisymmetric linear combination of  two orthogonal eigenstates, $\vert \psi_{h^1 j^1}\rangle$ and $\vert \psi_{h^2 j2}\rangle$, $(h^1,j^1)\neq (h^2,j^2)$, with almost identical energies
$\vert E^1 - E^2 \vert \leq 2^{-n}$, where
$E^1=\sum_a h^1_a$ and $E^2=\sum_a h^2_a$:
\begin{eqnarray}
\label{EQ:HDeg1X}
\HH^1(x)&=&\SPAN ( \nonumber\\
& &
\{  
\frac{1}{\sqrt{2}}\left( \vert \psi_{h^1j^1} \rangle  \vert \psi_{h^2j^2} \rangle - 
\vert \psi_{h^2j^2}\rangle \vert \psi_{h^1j^1} \rangle\right)\nonumber\\
& &\ :\   (h^{1,2},j^{1,2}) \in F,
(h^1,j^1)\neq (h^2,j^2)
, \nonumber\\
& &\quad
\vert E^1 - E^2 \vert \leq 2^{-n}
\}
\nonumber\\
& )&
\end{eqnarray} 
and denote by $\HH^0(x)$ the orthogonal subspace:
\begin{eqnarray}
\label{EQ:HDeg0X}
\HH^0(x)&=& \SPAN ( \nonumber\\
& &
\{
\vert \psi_{hj} \rangle  \vert \psi_{hj} \rangle
\vert   (h,j) \in F \}
\nonumber\\
& \cup & 
\{
\frac{1}{\sqrt{2}}\left( \vert \psi_{h^1j^1} \rangle  \vert \psi_{h^2j^2} \rangle + 
\vert \psi_{h^2j^2}\rangle \vert \psi_{h^1j^1} \rangle\right)\nonumber\\
& & \ :\   (h^{1,2},j^{1,2}) \in F, (h^1,j^1)\neq (h^2,j^2)
\}\nonumber\\
& \cup &
\{
\frac{1}{\sqrt{2}}\left( \vert \psi_{h^1j^1} \rangle  \vert \psi_{h^2j^2} \rangle - 
\vert \psi_{h^2j^2}\rangle \vert \psi_{h^1j^1} \rangle\right)\nonumber\\
& & \ :\   
 (h^{1,2},j^{1,2}) \in F,
(h^1,j^1)\neq (h^2,j^2)
,\nonumber\\
& & 
\quad 
\vert E^1 - E^2 \vert > 2^{-n}
\}
\nonumber\\
& )&
\end{eqnarray} 
\end{deff}

\begin{thm} The pair of subspaces $(\HH^1(x),\HH^0(x))$ defined in Definition$~\ref{Prob:klocalH}$ belong to $(1,0)$-total functional gap $\QMA$: 
\begin{equation}
(\HH^1(x),\HH^0(x)) \in \TFgapQMA (1,0)\ .
\end{equation}
\end{thm}

\begin{proof}

{\bf Quantum verification procedure.}

We first describe $Q$.

Carry out a SWAP test. Reject if the SWAP test outputs "Symmetric".

If the SWAP test outputs 
"Antisymmetric",
measure all the $H_a$ on the first $n$ qubits and on the second $n$ qubits to obtain the eigenvalues $h^1=(h^1_1, h^1_2,...,h^1_A)$ and $h^2=(h^2_1, h^2_2,...,h^2_A)$. Compute  the energies $E^1=\sum_a h^1_a$ and $E^2 =\sum_a h^2_a$. 
Reject if $\vert  E^1 -   E^2\vert >  2^{-n}$ 
and accept if $\vert E^1 -  E^2\vert \leq  2^{-n}$. (Recall that according to our hypothesis, the difference of energies can be computed exactly using an efficient classical algorithm.)

{\bf $Q$ is a gapped $(1,0)$-procedure.}

Note that the states enumerated in Equations~\eqref{EQ:HDeg1X} and \eqref{EQ:HDeg0X} form a basis of the space of $2n$ qubits. Consequently we have $\HH_{2n}=\SPAN(\HH^1(x) \cup \HH^0(x))$.

It is straighforward to check that $Q$ leaves all the states enumerated in Equations~\eqref{EQ:HDeg1X} and \eqref{EQ:HDeg0X}  invariant, and accepts with probability $1$ on the states in  Equation~\eqref{EQ:HDeg1X}, and rejects with probability $1$ on the states in  Equation~\eqref{EQ:HDeg0X}. Therefore $Q$ is a gapped $(1,0)$-procedure, and the states enumerated in Equations~\eqref{EQ:HDeg1X} and \eqref{EQ:HDeg0X} form an eigenbasis of $Q$ for $x$.

{\bf $Q$ is a total procedure.}

To prove that $Q$ is total, we show that $\HH^1(x)$ is non empty for all $x$.

Since there are $2^n$ states $\vert \psi_{hj}\rangle$ and their energies lie in the interval $[0,1]$, 
by the pigeonhole principle, there are at least two different states,
$\vert \psi_{hj}\rangle$ and $\vert \psi_{h'j'}\rangle$ with $(h,j)\neq (h',j')$,
such that the corresponding
energies  differ by at most $ 2^{-n}$. The quantum verification procedure therefore accepts with probability $1$ on the antisymmetric linear combination of these states.

\end{proof}

\subsubsection{Multiple copies of eigenstates of commuting $k$-local Hamiltonian.}\label{MultCopies}

The quantum no--cloning principle suggests another type of problem, namely producing several copies of a state that has certain properties. In order to translate this requirement into a quantum verification procedure it must be possible to verify these properties efficiently. 
We illustrate this in the  case of a commuting $k$-local Hamiltonian $H=\sum_{a=1}^A H_a$. The required property is that the states be joint eigenstates of all the $H_a$'s with the same eigenvalues.

Note that creating a single  joint eigenstate of the $H_a$'s (with random eigenvalues) is easy: take the completely mixed state (half of a maximally entangled state) and measure all the $H_a$ operators on the state. 
To create two identical copies, we can try the following procedure: start with the maximally entangled state $\vert \phi^+\rangle =2^{-n/2} \sum_{i=0}^{2^n-1} \vert i\rangle_1 \vert i\rangle_2$ (which can be efficiently produced). Now measure the $H_a$'s on the first system. Denote by $h=(h_1,...,h_n)$ the measured eigenvalues. If the corresponding eigenspace is one-dimensional, the state after the measurement is $\vert \psi_h\rangle_1 \vert \psi_h^*\rangle_2$, where $ \vert \psi^*\rangle$ denotes the complex conjugate of the state 
$\vert \psi \rangle$  in the standard basis. (If the corresponding eigenspace is degenerate with degeneracy $J_h$, the state after the measurement is $J_h^{-1/2}\sum_{j=1}^{J_h} \vert \psi_{hj}\rangle_1 \vert \psi_{hj}^*\rangle_2$ 
where $\{ \vert \psi_{hj}\rangle ; j=1,..., J_h\}$ is an orthonormal basis of the eigenspace with eigenvalues $h$). Thus if the $H_a$'s are real in the standard basis, we can efficiently create two identical eigenstates. But we do not know an efficient procedure to create two identical eigenstates when the $H_a$'s are complex, nor do we know of an efficient procedure to create three identical eigenstates when the $H_a$ are real.

These remarks lead to the following problem:

\begin{deff} \label{Prob:MCopies}
{\bf Multiple copies of eigenstates of commuting $k$-local Hamiltonian.}
Denote by $x$ the classical description of a
 commuting $k$-local Hamiltonian acting on the space of $n$ qubits $\HH_n$, $H=\sum_{a=1}^A H_a$,
 and where by hypothesis 
  each $k$-local term can be measured with zero error in polynomial time using a quantum computer.
 
 Denote by $\HH^1(x)$  the subspace of the space of $3n$ qubits $\HH_{3n}$ spanned by 
the products of states with the same eigenvalues $h$:
\begin{eqnarray}
\label{EQ:HCopies1X}
\HH^1(x)&=&\SPAN (
\nonumber\\
& &
\{
\vert \psi_{h j} \rangle  \vert \psi_{h j'} \rangle  \vert \psi_{h j''} \rangle\nonumber\\
& &\ :\    h \in G,
 j,j',j''\in J_h
 \}
\nonumber\\
& &)
\end{eqnarray} 
and denote by $\HH^0(x)$ the orthogonal subspace:
\begin{eqnarray}
\label{EQ:HCopies0X}
\HH^0(x)&=& \SPAN ( \nonumber\\
& &
\{ \vert \psi_{hj} \rangle  \vert \psi_{h'j'} \rangle\vert \psi_{h'''j''} \rangle
\nonumber\\
& &\quad 
\ :\    (h,j), (h',j')  ,(h'',j'') \in F,
\nonumber\\
& &\quad  h\neq h' {\rm\ OR\ } h' \neq h'' {\rm\ OR\ }   h'' \neq h\}
\nonumber\\
& &)\ .
\end{eqnarray} 
\end{deff}

(For definiteness we have considered the case where we request $3$ copies of the eigenstates. The case where the $H_a$ are complex and we request $2$ copies can be treated in the same way).

\begin{thm} The pair of subspaces $(\HH^1(x),\HH^0(x))$ defined in Definition$~\ref{Prob:MCopies}$ belong to $(1,0)$-total functional gap $\QMA$: 
\begin{equation}
(\HH^1(x),\HH^0(x)) \in \TFgapQMA (1,0)\ .
\end{equation}
\end{thm}

\begin{proof}

{\bf Quantum verification procedure.}

We first describe $Q$.

Measure all the $H_a$ on qubits $1,...,n$, on qubits $n+1,...,2n$, and on qubits $2n+1,...,3n$.
Accept if the outcomes $h=(h_1,...,h_A)$ are equal. Otherwise reject.

{\bf $Q$ is a gapped $(1,0)$-procedure.}

Note that the states enumerated in Equations~\eqref{EQ:HCopies1X} and \eqref{EQ:HCopies0X} form a basis of the space of $3n$ qubits. Consequently we have $\HH_{3n}=\SPAN( \HH^1(x) \cup \HH^0(x))$.

It is straighforward to check that $Q$ leaves all the states enumerated in Equations~\eqref{EQ:HCopies1X} and \eqref{EQ:HCopies0X}  invariant, and accepts with probability $1$ on the states in  Equation~\eqref{EQ:HCopies1X}, and rejects with probability $1$ on the states in  Equation~\eqref{EQ:HCopies0X}. Therefore $Q$ is a gapped $(1,0)$-procedure, and the states enumerated in Equations~\eqref{EQ:HDeg1X} and \eqref{EQ:HDeg0X} form an eigenbasis of $Q$ for $x$.

{\bf $Q$ a total procedure.}

Since $\{\vert \psi_{hj} \rangle  \}$ is a basis of $\HH_{n}$, 
 $\HH^1(x)$ is non empty for all $x$. In fact 
 ${\mbox{\rmfamily\rm{Dim}}}\left( \HH^1(x)\right) \geq 2^n$.

 \end{proof}

\subsection{Quantum Lov\'asz Local Lemma}

The Quantum Lov\'asz Local Lemma ($\QLLL$) introduced in \cite{AKS12} provides conditions under which the quantum 
$k$-SAT problem is satisfiable. 
The satisfiability conditions were extended in 
\cite{SMLM16} and \cite{HLSZ19}.

As an example  we give the following result taken from \cite{AKS12}:
Let $\{\Pi_1, . . . ,\Pi_m\}$ be a $k$-QSAT instance where all projectors have rank at most $r$. If every
qubit appears in at most $D = 2^k/(e\cdot r \cdot k)$ 
projectors, then the problem is satisfiable. For our purposes
we will call the hypothesis of this statement the $\QLLL$ condition.

A Constructive Quantum Lov\'asz Local Lemma provides conditions under which the frustration free state can be efficiently constructed by a quantum algorithm, i.e. is in $\FBQP$. Initial results used commutativity of the constraints \cite{SCV13,SA15}. This condition was dropped in \cite{GS16} which provides a constructive algorithm under a uniform gap constraint defined as follows:
let $\epsilon=1/q(n)$ for some polynomial $q(n)$, then
for any subset $S$ of the constraints the gap of $H_S=\sum_{i\in S} \Pi_i$ is greater than $\epsilon$, where the gap is the difference between the two smallest eigenvalues of $H_S$. 
Note that there is no known  efficient quantum algorithm that can check whether the uniform gap constraint is satisfied.

It is not known how the constructive algorithm of \cite{GS16} works when the uniform gap condition does not hold. It may be that it always outputs a state close to the ground state. It may also be that it sometimes outputs a state far from the ground state. If the latter is true, then this gives rise to an interesting problem in $\TFQMA$.

\begin{deff}
{\bf Quantum verification procedure for ground state energy under QLLL conditions.}
\label{procedureQLLL}
Denote by $x$ the classical description $\{\Pi_1, . . . ,\Pi_m\}$ of a $k$-QSAT instance satisfying the QLLL condition.
Denote by $H=m^{-1}\sum_{i=1}^m \Pi_i$ the Hamiltonian obtained by summing all the projectors, rescaled to have eigenvalues in the interval $[0,1]$.
 
Denote by $Q_{QLLL}$ the following quantum verification procedure:

On input 
{$(x, \vert \psi \rangle)$},
apply to $\vert \psi \rangle$  the  phase estimation algorithm \cite{S94,K95,CEMM98} for the unitary operator $U= \exp(i \pi H)$  to
$\ell(n)$ bits of precision, where $\ell \in \poly$.

In order to implement the unitary operator $U=\exp(i \pi H)$ and its powers in the phase estimation algorithm, use the algorithm of \cite{B17} that efficiently realises Hamiltonian simulation with exponentially small error.  Fix the error of the Hamiltonian simulation so that the  error made during  the phase estimation algorithm is at most $1/h(n)$, where $h\in\poly$ . 

Denote by $\tilde \phi$ the estimated phase.

Accept if $\tilde \phi =0$. Otherwise reject.

\end{deff}

Recall that on an eigenstate of $U$, $U\vert \psi_\phi \rangle = e^{i 2 \pi \phi} \vert \psi_\phi \rangle$, the phase estimation algorithm yields an $l$ bit approximation of $\phi\in [0,1)$. If the eigenphase $\phi$ is a multiple of $2^{-l}$ (i.e. if $\phi$ can be written exactly in binary using $l$ bits), then the phase estimation algorithm will yield the exact value of $\phi$ with probability $1$.

Note that in Definition$~\ref{procedureQLLL}$  we take $U=\exp(i \pi H)$ so that the eigenphases $\phi$ of $U$ lie in the interval $[0,1/2)$. This ensures that we cannot mistake the large eigenvalues with the small ones (which would be the case if we had taken $U=\exp(i 2\pi H)$).
Further details on the phase estimation algorithm can be found in  \cite{CEMM98}, see also the discussion in Section$~\ref{subsec:QFT}$.

\begin{thm}
 The quantum verification procedure $Q_{QLLL}$ described in  
Definition$~\ref{procedureQLLL}$ is a $(1-\frac{1}{h})$--Total Quantum Verification Procedure.
Furthermore the ground states of $H$ belong to
$R_{Q_{QLLL}}^{\geq 1-1/h}(x)$; and 
 all eigenstates of $H$  with energy $E\geq 2^{-\ell+1}$ belong to $R_{Q_{QLLL}}^{\leq 1/2+1/h}(x)$.
\label{Thm:QLLL2}
\end{thm}

\begin{proof}

 Denote by $\vert \psi_{Ej}\rangle$ the eigenstates of $H$ with energy $E$: $H
 \vert \psi_{Ej}\rangle = E \vert \psi_{Ej}\rangle$, where $j\in J_E$ labels orthogonal energy eigenstates with the same energy $E$.
 Recall that $0\leq E \leq 1$.
 By the $\QLLL$ conditions, the Hamiltonian has at least one frustration free state, i.e. a state with energy $0$. We denote these ground states $\vert \psi_{0j}\rangle$, $j\in J_0$.

First let us neglect the error made in the Hamiltonian simulation. 

The phase estimation algorithm acting on state $\vert \psi_{Ej}\rangle$ will output $\tilde \phi$, which is an $\ell$ bit estimate of  $E/2$.

Recall that if $E/2$ is an integer multiple of $2^{-\ell}$,
then $\tilde \phi = E/2$ with probability $1$.
As a consequence the ground states $\vert \psi_{0j}\rangle$ will {be accepted} with probability $1$.

Taking into account the error in the Hamiltonian simulation, the probability that the quantum verification procedure accepts on  
 $\vert \psi_{0j}\rangle$ is at least $1-\frac{1}{h}$. Hence   procedure $Q_{QLLL}$
is a $(1-\frac{1}{h})$--total quantum verification procedure and the ground states $\vert \psi_{0j}\rangle$ belongs to
$R_{Q}^{\geq 1-1/h}(x)$.

Let us now consider the probability that the procedure $Q_{QLLL}$ accepts on the other eigenstates $\vert \psi_{Ej}\rangle$. Once again we first neglect the error made in the Hamiltonian simulation.
 It follows from the analysis of  
\cite{CEMM98}  that the probability that $\tilde \phi$ differs from $E/2$ by more than $2^{-\ell}$ is less than $1/2$. Hence all eigenstates with energy $E\geq 2^{-\ell+1}$ will accept with probability less or equal than $1/2$.

Taking into account the error in the Hamiltonian simulation, the probability that the quantum verification procedure accepts on  
 $\vert \psi_{Ej}\rangle$ with $E\geq 2^{-\ell+1}$  is at most $1/2+1/h$. 
 \end{proof}
 
Note that there may exist  eigenstates with energy $2^{-\ell+1}\geq E >0$. We do not know what is the acceptance probability of $Q_{QLLL}$ on these eigenstates.

Note that the 
procedure of Theorem \ref{Thm:FQMASR}  allows us to change the bounds $1-1/h$ and $1/2+1/h$ that appear in the statement of Theorem \ref{Thm:QLLL2}, for instance to $2/3$ and $1/3$. However a detailed  analysis is complicated by the fact that we do not know the  eigenbasis of $Q_{QLLL}$. (If the Hamiltonian simulation did not induce any error, then the eigenbasis of $Q_{QLLL}$ would consist of the energy eigenstates $\vert \psi_{Ej}\rangle$. The error in the Hamiltonian simulation modifies the eigenbasis slightly.)

Note that the  classical analogue of the  problem 
based on the quantum verification procedure $Q_{QLLL}$ 
is in $\FBPP$ (the functional analog of $\BPP$), as there exist efficient randomized classical algorithms to find a satisfying assignment when the Lov\'asz Local Lemma conditions are satisfied\cite{M09,MT10}.

\subsection{Quantum money based on knots.}

Public key quantum money was introduced in \cite{A09}.
 Here we show how the scheme of \cite{FGHLS12} in which the quantum money consists of coherent superposition of (representations of) knots induces a problem in $\TFQMA$.

We first recall that any knot can be represented by a grid diagram $G$.
We denote by $D(G)$ the size of the grid diagram.
Any grid diagram $G$   can be encoded by two disjoint permutations $\Pi_X$ and $\Pi_O$ of $D(G)$ elements. 
We denote by 
\begin{equation}
\vert G\rangle =\vert D(G),\Pi_X,\Pi_O\rangle
\label{Eq:Gform}
\end{equation}
 a quantum encoding of such a grid diagram. The one-variate Alexander polynomial  $A(G)$ 
can be efficiently computed from the representation $G$ of a knot \cite{A28}.

In \cite{FGHLS12} it is proposed that the following states, labeled by grid diagrams $G$, can be used as quantum money
\begin{equation}
\vert \$_{G}\rangle =\sum_{
\substack{G' : ~ 
2\leq D(G') \leq 2D(G), \\
 ~A(G')=A(G)}}
  \frac{\sqrt{q(D(G'))}}{\sqrt{N} }\ 
\vert G' \rangle,
\label{QMoneyKnot}
\end{equation}
where $D(G)$ is the dimension of the grid diagram $G$; $A(G)$ is the   Alexander polynomial of the corresponding knot;  
the superposition is over grid diagrams $G'$ of dimension between $2$ and $2D$ with the same Alexander polynomial $A(G')=A(G)$;
$N$ is a normalisation factor;
 $q(d')$ is the following quasi--Gaussian distribution over 
grid diagram dimensions between $2$ and $2D$: 
$q(d') = \left \lceil{y(d')/y_{\min}}\right \rceil$, where $y(d')= 
\frac{1}{d'! \left[ \frac{d'!}{e}\right]}  \exp \left(-(d'-D )^2/ 2 D\right)$, for $2 \leq d' \leq 2D$,
with $y_{\min}$  the minimum value of $y(d')$ for $2 \leq d' \leq 2D$, and where
for a positive real number $x$ we denote by $\lceil x \rceil$  the smallest integer which is at least $x$, and we set $[x] = \lceil x - 1/2 \rceil$.

Note that one does not know of an efficient procedure to check if a polynomial is an Alexander polynomial associated to a knot, nor of an efficient algorithm which, given an Alexander polynomial, finds the associated knot. For this reason the input to the following procedure is a grid diagram $G$ and a quantum state.

\begin{deff}\label{Prob:Knot}
{\bf Quantum verification procedure for quantum money based on knots.}
Denote by $Q_\$ $  the quantum verification procedure described in \cite{FGHLS12}, which for completeness we recall briefly. 

On input $(G,\vert \phi \rangle)$ carry out the following steps:
\begin{enumerate}
\item 
Verify that $\vert \phi \rangle$ is a superposition of basis vectors that validly encode grid
diagrams, i.e. that it has the form Eq. \eqref{Eq:Gform}.
If this is the case then move on to step 2, otherwise reject.
\item
Measure the Alexander polynomial on $\vert \phi \rangle$. If this is measured
to be $A(G)$ then continue on to step 3. Otherwise, reject.
\item
Measure the projector onto grid diagrams with dimensions in the range $\left[D(G)/2,3D(G)/2\right]$.  If you obtain +1 then continue  to step 4. Otherwise, reject. 
\item
Apply the Markov chain verification algorithm described in \cite{FGHLS12}.
If $\vert \phi \rangle$ passes this step, accept. Otherwise, reject. 

This is the crucial step that checks that the state is a coherent superposition of knots which can be mapped one into the other by elementary grid moves, that is elementary moves that map a knot onto an equivalent knot.

\end{enumerate}
\end{deff}

\begin{thm}
The quantum verification procedure $Q_{\$}$ described in  Definition$~\ref{Prob:Knot}$ is a $(1-C\exp(-D(G)/2))$--Total Quantum Verification Procedure, for some positive constant $C$.
Furthermore the  states Equation~\eqref{QMoneyKnot} belong to
$R_{Q_\$}^{\geq 1-C\exp(-D(G)/2)}(G)$.
\label{Thm:Knot}
\end{thm}

\begin{proof}
Consider the action of $Q_{\$}$  on input $(G, \vert \$_{G}\rangle )$.
Steps 1 and 2 succeed with probability $1$. 
Step 3 succeeds with probability $1-\delta$ where $\delta$ is  approximately given by $\exp(-D(G)/8)$.

Note that the unnormalised state after step 3 can be written

\begin{equation}
(1-\delta)\vert \$_{G}\rangle  + \vert \$_{G}^\perp\rangle
\label{eq:distState}
\end{equation}
where $\vert \$_{G}^\perp\rangle$ is orthogonal to $\vert \$_{G}\rangle $ and has norm 
$\langle \$_{G}^\perp \vert \$_{G}^\perp\rangle = \delta (1- \delta)$.

Given as input a state of the form
$\vert \$_{G}\rangle$, step 4 succeeds with probability 1. However because the  state has been distorted at step 3 (see Equation \eqref{eq:distState}), 
on input $(G, \vert \$_{G}\rangle )$ step 4 of $Q_{\$}$
succeeds with slightly reduced probability lower--bounded by $1-O(\sqrt{\delta})$.

Hence   procedure $Q_{\$}$
is a $1-O(\sqrt{\delta})$--total quantum verification procedure and the  state $ \vert \$_{G}\rangle$ belongs to
$R_{Q}^{\geq 1-O(\sqrt{\delta})}(G)$.

Using the inequality $\sqrt{\delta} \leq C\exp(-D(G)/2)$ for some positive constant $C$  provides the statement in the proof.
\end{proof}

It is not known what other states will pass the above quantum verification procedure. It is conjectured, see discussion in \cite{FGHLS12}, that 
quantum computers cannot efficiently produce states that pass the above quantum verification procedure.

\section{Relativized Problems}\label{SEC:Rel}

\subsection{Introduction}

In this section we give problems in which the quantum computer has access to an oracle. The complexity is counted as the complexity of the quantum algorithm, including the number of calls to the oracle which each count as one computational step.

A quantum oracle is  an infinite sequence of unitary transformations $U=\{ U_n\}_{n\geq 1}$.  We assume that each $U_n$ acts on $p(n)$ qubits for some  $p\in \poly$.   We  assume that given an $n$-bit string as input, a quantum algorithm calls only $U_n$, not $U_m$ for any $m\neq n$.
When there is no danger of confusion, we will refer to $U_n$ simply as $U$.

We now describe how one makes a call to the oracle.  Assume a quantum computer’s state has the form
\begin{equation}
\vert \Phi \rangle = \sum_{z}\sum_{b\in \{-1,0,1\}} \alpha_{z,b} \vert z\rangle 
\vert b\rangle \vert \phi_{z,b}\rangle  
\end{equation}
where $\vert z\rangle$ is a basis of the workspace register, $\vert b\rangle$ is a control qutrit with basis $\{\vert -1\rangle, \vert 0\rangle, \vert +1\rangle\}$, and $\vert \phi_{z,b}\rangle$ is a $p(n)$-qubit answer register.  Then to “query $U_n$” means to apply the
following unitary transformation
\begin{equation}
\vert \Phi \rangle 
\to 
 \sum_{z}\sum_{b\in \{-1,0,1\}} \alpha_{z,b} \vert z\rangle 
\vert b\rangle U^{b}\vert \phi_{z,b}\rangle  
 \ , 
\end{equation}
where we have assumed that if we can apply $U$, then we can also apply controlled--$U$  and controlled--$U^{-1}$.

Let $C$ be a quantum complexity class, and let $U=\{ U_n\}_{n\geq 1}$ be a quantum oracle.   Then by $C^U$, we  mean the class of problems solvable by a $C$ machine that, given an input of length $n$, can query $U_n$ at unit cost as many times as it likes.

\subsection{Finding a marked state}

We first give a very simple oracle, which is the basis of Grover's algorithm \cite{G96,G97} with respect to which we have a separation between FBQP and TFQMA. See \cite{BBGV97,AK07} for previous use of this oracle in separating complexity classes.

\begin{oracle}{\bf Marking a state.}
Let $\{\vert \psi_n \rangle\in \HH_n; n\in \mathbb{N}\} $ be  a family of states chosen uniformly at random from the Haar measure.
We denote  by $A=\{A_n\}$ the oracle acting on $n+1$ qubits that marks the $n$ qubits state $\vert \psi_n \rangle\in \HH_n$:
\begin{eqnarray}
A_n\vert a\rangle \vert \psi_n \rangle &=& \vert a \oplus 1\rangle \vert \psi_n \rangle\ ,\nonumber\\
A_n\vert a\rangle \vert \phi \rangle &=& \vert a \rangle \vert \phi \rangle\quad \forall \vert \phi \rangle \perp \vert \psi_n \rangle 
\end{eqnarray}
where $a\in\{0,1\}$. 
\end{oracle}

\begin{deff} \label{Prob:MarkedState}
{\bf Finding a marked state.} Given  oracle $A$ and the corresponding family of states
$\{\vert \psi_n\rangle \} $, denote by $\HH^1(n)$ the space spanned by the state $\vert \psi_n \rangle $ and denote by 
$\HH^0(n)$ the orthogonal space:
\begin{eqnarray}
\label{EQ:MarkedState}
\HH^1(n)&=& \SPAN ( \{ \vert \psi_n \rangle \})\ ,
\nonumber\\
\HH^0(n)&=&
\SPAN ( \{\vert \phi \rangle \ : \ 
\vert \phi \rangle \perp \vert \psi_n \} )\ .
\end{eqnarray} 
\end{deff}

\begin{thm} 
\label{ThmOracleSeparation}
The pair of subspaces $(\HH^1(x),\HH^0(x))$ defined in Definition$~\ref{Prob:MarkedState}$ belong to $(1,0)$-total functional gap $\QMA^A$, but are not in  $\FBQP^A$:
\begin{eqnarray}
(\HH^1(x),\HH^0(x)) &\in& \TFgapQMA^A (1,0)\\
(\HH^1(x),\HH^0(x)) &\notin& \FBQP^A\ .
\end{eqnarray}
\end{thm}

\begin{proof}
{\bf Quantum verification procedure.}

We first describe the quantum verification procedure $Q$.

The value of the classical input $x$ is irrelevant, only its length $n$ is used. On input $(n, \vert \chi\rangle)$, append to   $\vert \chi\rangle)$ a single qubit in state $\vert 0\rangle$ to obtain the state $\vert 0\rangle\vert \chi\rangle$; act with $A$ on this state; measure the first qubit; accept if the measurement result is $1$ and reject if the measurement result is $0$.

{\bf $Q$ is a gapped $(1,0)$-total procedure.}

It is immediate to show that $Q$  accepts with probability $1$ on the states in $\HH^1(n)$ and accepts with probability $0$ on the states in $\HH^0(n)$. 
Furthermore  $\HH^1(n)$ is non empty for all $x$. Therefore $Q$ is a $(1,0)$ gapped total quantum verification procedure, and $(\HH^1(n),\HH^0(n)) \in \TFgapQMA^A (1,0)$.

{\bf Hardness in $\FBQP^A$.}

It is well known that finding  a marked state in a Hilbert space of dimension $d$  
requires $\Theta(d^{1/2})$ queries to the oracle. The lower bound follows from arguments in \cite{BBGV97}, and the upper bound is given by Grover's algorithm\cite{G96,G97}. Since $d=2^n$ a quantum computer will need  $\Theta(2^{n/2})$ operations to find the marked state.

\end{proof}

\subsection{Group Non--Membership}

Black-box groups, in which group operations are performed by an oracle $B$, were introduced by Babai and Szemer\'edi  in \cite{BS84}. In this model subgroups are given by a list of generators. It was shown in \cite{BS84} that for such subgroups, Group Membership belongs to $\NP^B$, i.e. there exists a succinct classical certificate for membership. Subsequently, by extending the oracle $B$ to the quantum setting,
Watrous \cite{Watrous} showed  that Group Non-Membership is in $\QMA^B$, i.e. there exists a succinct quantum certificate for non-membership. Consequently, as we show below, the general question of Group (Non-)Membership belongs to $\TFQMA^B$.  (Note that \cite{AK07} provides evidence that the certificate for Group Non-Membership could be classical, in which case Group (Non-)Membership would belong to  $\TFCQMA^B$).

\begin{oracle}{\bf Black-box groups.}
We use  Babai and Szemer\'edi's model of black-box groups with unique encoding \cite{BS84}, adapted to the quantum context.
In this model we know how to multiply
and take inverses of  elements of the group, but we don't know anything else about the group.

More precisely, 
let $\{G_n\}$ be a family of groups, with $\vert G_n \vert \leq 2^n$.
Each element $x\in G_n$ is represented by a randomly chosen classical label $ l(x)\in \{0,1\}^n$, to which we associate a quantum state $\vert  l(x)\rangle $ (the label $l(x)$ written in the computational basis).
We denote by $B=\{B_n\}$ the family of  oracles that perform the group operations as follows:

If  the state of the quantum computer is 
\begin{equation}
\vert \psi\rangle = \sum_{x,y\in G_n}\sum_z \psi_{xyz} \vert l(x)\rangle \vert l(y)\rangle \vert z\rangle,
\label{GroupOracle1}
\end{equation}
where $ \vert z\rangle$ is some workspace, then the oracle acts as
\begin{equation}
B_n\vert \psi\rangle = \sum_{x,y\in G_n}\sum_z  \psi_{xyz} \vert l(x)\rangle \vert l(yx^{-1})\rangle \vert z\rangle.
\label{GroupOracle2}
\end{equation}

We suppose that the representation of the unit element $\vert l(e)\rangle$ is known. The oracle can then be used to compute the inverse of an element (by inputting $\vert l(x)\rangle \vert l(e)\rangle$), and group 
multiplication (by first computing $\vert l(x^{-1})\rangle$, and then inputing  $\vert l(x^{-1})\rangle \vert l(y)\rangle$).

In addition we suppose that  the oracle can check that a register contains a valid label. One possibility is that if the inputs  are orthogonal to states of the form Equation~\eqref{GroupOracle1}, i.e. if the first two registers do not contain valid labels, then the oracle returns a standard error signal $\vert \perp \rangle$.

\end{oracle}

For simplicity of notation in the following we use interchangeably the notations $g$ and $l(g)$ for the group elements. The context will make clear which is used.

Fix the index $n$.
Suppose you receive as input the labels $l(g_1), ...,l(g_k)$ and $l(h)$ of group elements $g_1, ...,g_,h \in G_n$. 
Denote by $H=\langle g_1, ...,g_k\rangle $ the subgroup of $G_n$  generated by $g_1, ...,g_k$.
Group (Non-)Membership is the question: does $H$ contain $h$? 

 Babai and Szemer\'edi \cite{BS84} 
 showed that there exists a short classical certificate for $h \in H$, that we denote by $C(g_1, ...,g_k,h)$.
 The certificate is an efficient representation of $h$ as a product of the group elements $g_1, ...,g_k$ and their inverses $g_1^{-1}, ...,g_k^{-1}$, see \cite{BS84} for details.
 
 Watrous showed  \cite{Watrous} that for $h\notin H$ there exists a succinct quantum certificate.
\begin{equation}
\vert \psi_H\rangle =\frac{ 1}{\vert H \vert ^{1/2}}\sum_{x\in H} \vert l(x) \rangle \ .
\label{EqWatrousState}
\end{equation}

\begin{deff} {\bf Quantum verification procedure for group (non-)membership.}\label{Prob:GNM}
Given  oracle $B$, and the corresponding family of groups $\{G_n\}$,
let $x=(n,l(g_1), ...,l(g_k), l(h))$. 
Denote by $Q_{G(N)M}$ the following quantum verification procedure which on input $(x, \vert \psi\rangle)$ acts as:

\begin{enumerate}

\item  Measure the first qubit of the quantum input $\vert \psi\rangle$ in the standard basis. Denote by $\vert \psi'\rangle$ the remaining part of the quantum input.

\item  If the first qubit is $0$, then check whether $\vert \psi'\rangle= \vert C(g_1, ...,g_k,h)\rangle$ is a classical certificate certifying that $h\in H$. Accept if this is  the case, otherwise reject.

\item  if the first qubit is $1$, then on $\vert \psi'\rangle$ carry out  the quantum verification procedure for group non membership described in \cite{Watrous}. Accept or reject accordingly.

\end{enumerate}

\end{deff}

\begin{thm}
\label{Thm:GNM}
Given access to oracle $B$, the quantum verification procedure $Q_{G(N)M}$ described in Definition$~\ref{Prob:GNM}$ is an $ \frac{1}{2}$-total quantum verification procedure.
Furthermore,
\begin{enumerate}
\item If $h\in H$, then 
\begin{equation}
\vert 0 \rangle \vert C(g_1, ...,g_k,h)\rangle\in R^{1}_{Q_{G(N)M}}(x)
\end{equation}
for all valid certificates $C(g_1, ...,g_k,h)$ that $h\in H$;
 and all states with the first qubit set to $1$ reject with high probability:
\begin{equation}
\{\vert 1 \rangle \vert \psi'\rangle 
\}\in 
R^{\leq 2^{-2n}}_{Q_{G(N)M}}(x);
\end{equation}
\item 
 If $h\notin H$, then the state
\begin{equation}
\vert 1 \rangle \vert \psi_H\rangle\in R^{\geq 1/2}_{Q_{G(N)M}}(x),
\end{equation}
  and all states with the first qubit set to $0$ reject with unit  probability:
\begin{equation}
\{\vert 0 \rangle \vert \psi'\rangle\}\in R^{0}_{Q_{G(N)M}}(x).
\end{equation}
\end{enumerate}
\end{thm}

\begin{proof}
If the input has the form $\vert 0\rangle \vert \psi' \rangle$, then the verification procedure is classical, and the probabilities of accepting is $1$ if $h\in H$ and the input is a valid classical certificate, otherwise the probability of accepting is $0$.

If the input has the form $\vert 1\rangle \vert \psi_H \rangle$ and $h\notin H$, then the probability that the quantum verification procedure for group non membership accepts is $1/2$ (see \cite{Watrous}).

If the input has the form $\vert 1\rangle \vert \psi' \rangle$, and $h\in H$, then the probability that the quantum verification procedure for group non membership accepts is upper bounded by $2^{-2n}$ (see \cite{Watrous}).

\end{proof}

\subsection{Problems based on QFT}
\label{subsec:QFT}

We consider here  problems for which the verification procedure is based on the efficiency of the Quantum Fourier Transform and the phase estimation algorithm \cite{S94,K95,CEMM98}. 

The Quantum Fourier Transform is based on a unitary that can be efficiently exponentiated. We will suppose below that this unitary is given by an oracle. 

Unitaries that can be efficiently exponentiated were studied in \cite{AA17} in the context of the time energy uncertainty. The only explicit example we are aware of where $U$ can be efficiently exponentiated but cannot be efficiently diagonalised is when $U$ is the time evolution of a commuting $k$-local Hamiltonian: $U=\exp (iH)$ with $H=\sum_a H_a$, where $H_a$ is $k$-local and the $H_a$ all commute. Therefore the problems below also apply in the case where the input $x$ is the classical description of such a commuting $k$-local Hamiltonian, and $U=\exp (iH)$.
If additional classes of unitaries that can be efficiently exponentiated but cannot be efficiently diagonalized are discovered, then this provides new $\TFQMA$ problems, which justifies using the present oracle based formulation.

\begin{oracle}{\bf Efficient exponentiation of unitaries.}
Let $\{ U_n \ :\  n\in \mathbb{N}\} $ be  a family of unitary matrices acting on $n$ qubits chosen uniformly at random from the Haar measure.
We denote  by $C=\{C_n\}$ the oracle which implements the transformations $U_n$ and their powers as follows:
\begin{equation}
C_n \Big(  \vert k \rangle \vert \psi \rangle \vert \varphi \rangle\Big) 
= \vert k \rangle \Big(U_n^k \vert \psi \rangle \Big) \vert \varphi \rangle
\end{equation}
where $ \vert k \rangle$ is a classical register of $n$ bits, with $k\in\{0,...,2^n-1\}$, $ \vert \psi \rangle$ is a state of $n$ qubits, and $\vert \varphi \rangle$ is some workspace. 
\end{oracle}

We denote by $\phi\in[0,1)$ and $\vert \psi_{\phi\alpha}\rangle \in \HH_n$ the eigenphases and eigenstates of $U_n$: 
\begin{eqnarray}
U_n \vert \psi_{\phi\alpha} \rangle &=& e^{i2 \pi \phi} \vert \psi_{\phi\alpha} \rangle\nonumber,\\
\langle \psi_{\phi'\alpha'}\vert \psi_{\phi\alpha}\rangle&=&\delta_{\alpha' \alpha}\delta_{\phi' \phi},
\label{phialpha}
\end{eqnarray}
where $\alpha\in \N$ labels orthogonal states with the same eigenvalue.
(For simplicity of notation, we do not add an index $n$ to the states $\vert \psi_{\phi\alpha}\rangle$: it will be obvious from the context what size Hilbert space they belong to).

We denote by $S(n)$ the set of  couples $(\phi,\alpha)$ that satisfy Equation~\eqref{phialpha}:
\begin{equation}
S(n)=\{(\phi,\alpha) :  U_n \vert \psi_{\phi\alpha} \rangle = e^{i2 \pi \phi} \vert \psi_{\phi\alpha} \rangle \}
\end{equation}
and
we denote by 
$S^{2\, ,\,dis}(n)$ the set of distinct couples $((\phi,\alpha), (\phi',\alpha'))$:
\begin{eqnarray}
S^{2\, ,\,dis}(n)&=&\{
(\phi,\alpha, \phi',\alpha')\in S(n) \times S(n)
\nonumber\\
& &\ :
(\phi < \phi') \rm{\ OR \ }
(\phi = \phi' \rm{\ AND \ } \alpha < \alpha')\}\ .\nonumber
\end{eqnarray}

We denote by $\HH^{sym}(n)$ the symmetric space
\begin{eqnarray}
\HH^{sym}(n) &=& \SPAN ( \nonumber\\
& & \{  \vert \psi_{\phi \alpha} \rangle\vert \psi_{\phi \alpha} \rangle
: (\phi,\alpha) \in S  \}
\nonumber\\
&  \cup & \{ 
\frac{1}{\sqrt{2}} \left( \vert \psi_{\phi \alpha} \rangle \vert \psi_{\phi'\alpha'} \rangle
+ \vert \psi_{\phi' \alpha'} \rangle \vert \psi_{\phi\alpha} \rangle \right)  \nonumber\\
& &
: (\phi,\alpha,\phi',\alpha') \in  S^{2\, ,\,dis}\} \nonumber\\
& & )\ ,
\end{eqnarray}
and by $\HH^{anti}(n)$ the antisymmetric space
\begin{eqnarray}
\HH^{anti}(n) &=& \SPAN ( \nonumber\\
&  & \{ 
\vert \psi^A_{\phi \alpha\phi' \alpha'}\rangle 
: (\phi,\alpha,\phi',\alpha') \in  S^{2\, ,\,dis} \}
\nonumber\\  & & ) \ ,
\end{eqnarray}
with
\begin{eqnarray}
\vert \psi^A_{\phi \alpha\phi' \alpha'}\rangle &=&
\frac{ \vert \psi_{\phi \alpha} \rangle \vert \psi_{\phi'\alpha'} \rangle
- \vert \psi_{\phi' \alpha'} \rangle \vert \psi_{\phi\alpha} \rangle }{\sqrt{2}}
\label{eq:PsiUAnti}
\end{eqnarray}
the antisymmetric states.

We denote by $2\pi \dd(\phi, \phi')$ is the distance on the unit circle  between the angles $2\pi \phi$ and $2\pi\phi'$:
\begin{equation}
\dd(\phi, \phi')= 
\min \{\vert \phi-\phi'\vert, 1 - \vert \phi -\phi' \vert\}\ .
\end{equation}

The following problem is a generalisation of the problem based on Definition \ref{Prob:klocalH}
 to the case where the input is an oracle implementing unitary transformations, rather than by a commuting $k$-local Hamiltonian.

\begin{deff} \label{Prob:AlmostDegU}
{\bf Almost Degenerate Eigenspace of $U$.}
Given oracle $C$, 
denote by
$S^{\leq 2^{-n}}\subseteq S\times S$ the set of neighbouring eigenvalues of $U_n$, and by $\HH^{\leq 2^{-n}}$ the corresponding subspace of $\HH^{anti}(n)$: 
\begin{eqnarray}
S^{\leq 2^{-n}}&=& \{ ( (\phi,\alpha), (\phi',\alpha') )\in S\times S \nonumber\\
& &  : \dd(\phi,\phi')\leq 2^{-n}, 
(\phi',\alpha')\neq (\phi,\alpha) \} \nonumber\\
\HH^{\leq 2^{-n}} &=&  \SPAN ( \nonumber\\
&  & \{ 
\vert \psi^A_{\phi \alpha\phi' \alpha'}\rangle 
: ((\phi,\alpha), (\phi',\alpha')) \in  S^{\leq 2^{-n}} \} \nonumber\\
& & )\ ;
\end{eqnarray}
and denote by $S^{>9/2^{n}}\subseteq S\times S$ the set of non-neighbouring eigenvalues of $U_n$, and by $\HH^{>9/2^{n}}$ the corresponding subspace of $\HH^{anti}(n)$: 
\begin{eqnarray}
S^{>9/2^{n}}&=& \{ ( (\phi,\alpha), (\phi',\alpha') )\in S\times S \nonumber\\
& &  :
\dd(\phi,\phi')> 9/ 2^{n}, 
(\phi',\alpha')\neq (\phi,\alpha) \} \nonumber\\
\HH^{>9/2^{n}} &=&  \SPAN ( \nonumber\\
&  & \{ 
\vert \psi^A_{\phi \alpha\phi' \alpha'}\rangle
: ((\phi,\alpha), (\phi',\alpha')) \in S^{>9/2^{n}}\} \nonumber\\
& & )\ .
\end{eqnarray}
\end{deff}

\begin{thm}
\label{Thm:AlmostDegU}
Given access to oracle $C$, and given $n$, there exists a $2/3$-total quantum verification procedure $Q$ acting on $2n$ qubits such that the corresponding pair of relations 
$(\HH^{\geq 2/3\, ,\,C}_Q(n, \vert \psi \rangle), \HH^{\leq 1/3\, ,\,C}_Q(n, \vert \psi \rangle))\in \TFQMA^C(2/3,1/3)$ 
satisfy
\begin{eqnarray}
\HH^{\geq 2/3\, ,\,C}_Q(n, \vert \psi \rangle)&\supseteq &
\HH^{\leq 2^{-n}} \ ,\nonumber\\
\HH^{\leq 1/3\, ,\,C}_Q(n, \vert \psi \rangle))&\supseteq &
\SPAN ( \HH^{sym}(n) , \HH^{>9/2^{n}}  )\ .\nonumber\\
&&
\end{eqnarray}
\end{thm}

\begin{proof}

{\bf Quantum verification procedure.} 
We first describe $Q$, which we view as acting on two $n$ qubit states.

Step 1:  Carry out a SWAP test on the two $n$ qubit states. Reject if the SWAP test outputs "Symmetric"; proceed to Step 2 if the SWAP test outputs "Antisymmetric".

Step 2: Carry out the phase estimation algorithm on both  $n$ qubits states to $n$ bits of precision, obtaining two estimates $\hat \phi$ and $\hat \phi'$. Reject if $\dd(\hat \phi, \hat \phi') > 5 / 2^n$, otherwise accept.

{\bf Eigenbasis of $Q$.}

First note that the SWAP test leaves symmetric and antisymmetric spaces $\HH^{sym}(n)$ and $\HH^{anti}(n)$ invariant. Therefore the symmetric states, which all accept with probability $0$, constitute part of the eigenbasis of $Q$.

Second, recall that after phase estimation an eigenstate  of  $U_n$ is not modified, but the ancilla contains a superposition of estimates of the phase
\begin{equation}
\vert \psi_{\phi \alpha} \rangle\otimes \vert 0^n\rangle
\to
\vert \psi_{\phi \alpha} \rangle\otimes  \sum_{\hat \phi} c_{\phi \hat \phi}
\vert \hat \phi \rangle
\end{equation}
where $\hat \phi$ are  the $n$ bit estimates of the phase.
The probability of state $\vert \psi_{\phi \alpha} \rangle$ yielding estimate $\hat \phi$ is therefore
\begin{equation}
\Pr\left[\hat \phi \vert  \psi_{\phi \alpha}  \right] = \vert c_{\phi \hat \phi} \vert^2\ .
\end{equation} 
As a consequence, the probability that Step 2, acting on a linear superposition of antisymmetric states
\begin{equation}
\vert \psi \rangle = 
\sum_{(\phi,\alpha,\phi',\alpha') \in  S^{2\ dis}}
\gamma_{\phi \alpha \phi' \alpha'} \vert \psi^A_{\phi \alpha\phi' \alpha'}\rangle 
\label{eq:psiAnti}
\end{equation}
yields estimates $(\hat \phi, \hat \phi')$
 is
\begin{eqnarray}
&\Pr \left[\hat \phi, \hat \phi'\vert  \psi \right] =\quad \quad \quad &\nonumber\\
& 
\sum_{(\phi,\alpha,\phi',\alpha') \in  S^{2\, ,\,dis}}
\vert \gamma_{\phi \alpha \phi' \alpha'}\vert^2
\frac{  \vert  c_{\phi \hat \phi} \vert^2
\vert  c_{\phi' \hat \phi'} \vert^2
+
\vert  c_{\phi' \hat \phi} \vert^2
\vert  c_{\phi \hat \phi'} \vert^2
}{2}
\ .&
\nonumber\\
\label{eq:ProbaPsiAnti}
\end{eqnarray}
Since there are no interferences between the different antisymmetric states in the superposition,  the antisymmetric states are the other part of the eigenbasis of $Q$, see Theorem \ref{Thm:BlockStructure}.

{\bf Acceptance and rejection probability of antisymmetric states.}

Recall \cite{CEMM98}  that the phase estimation algorithm with $n$ bit of precision acting on an eigenstate $\vert \psi_{\phi \alpha} \rangle$ yields an estimated phase  with error bounded by
\begin{equation}\label{Eq:PhaseEst}
\Pr \left[ \dd( \phi ,\hat \phi) > \frac{k}{2^{n}}\right] < \frac{1}{2k-1}\ .
\end{equation}

For an antisymmetric state $\vert \psi^A_{\phi \alpha\phi' \alpha'}\rangle $ the quantum verification procedure $Q$ will yield two estimates for the phases $\hat \phi$ and $\hat \phi'$ with probability
\begin{eqnarray}
\Pr \left[\hat \phi, \hat \phi'\vert \psi^A_{\phi \alpha\phi' \alpha'} \right] &=&
\frac{  \vert  c_{\phi \hat \phi} \vert^2
\vert  c_{\phi' \hat \phi'} \vert^2
+
\vert  c_{\phi' \hat \phi} \vert^2
\vert  c_{\phi \hat \phi'} \vert^2
}{2}\nonumber\\
&=&
\frac{1}{2}
\left(
\Pr(\hat \phi \vert  \psi_{\phi \alpha}  ) \Pr(\hat \phi' \vert  \psi_{\phi' \alpha'}  )\right.
\nonumber\\
& &
\left. +
 \Pr(\hat \phi \vert  \psi_{\phi' \alpha'}  ) \Pr(\hat \phi \vert  \psi_{\phi' \alpha'}  )
\right)\nonumber\\
\label{eq:PsiA}
\end{eqnarray}

First we show that if $\vert \psi^A_{\phi \alpha\phi' \alpha'}\rangle
\in \HH^{\leq 2^{-n}} $, that is if
$\dd( \phi ,\phi') \leq 2^{-n}$, then $\Pr[ \dd( \hat \phi ,\hat \phi') )\leq 5/2^n] \geq  2/3$, i.e. the acceptance probability is greater or equal then $2/3$.

To this end we consider each term in Equation~\eqref{eq:PsiA} separately, for instance consider term $\Pr(\hat \phi \vert  \psi_{\phi \alpha}  ) \Pr(\hat \phi' \vert  \psi_{\phi' \alpha'}  )$.
Now use the triangle inequality to obtain
\begin{equation}\label{Eq:Triangle}
\dd(\hat \phi, \hat \phi') \leq \dd(\hat \phi, \phi) + \dd(\phi, \phi') + \dd(\phi', \hat \phi')\ .
\end{equation}
 Hence if $\dd(\hat \phi, \hat \phi') > 5/ 2^n$ and $\dd( \phi ,\phi') \leq 2^{-n}$, then either $ \dd(\hat \phi, \phi) >2/2^n$ or $ \dd(\phi', \hat \phi') >2/2^n$.
 From Equation~\eqref{Eq:PhaseEst} the probability of at least one of the later events occurring is less than $1/3$.
 Hence if $\dd( \phi ,\phi') \leq 2^{-n}$, then
 $\Pr[ \dd( \hat \phi ,\hat \phi') )> 5/2^n] <  1/3$, and consequently the  probability of the complementary event is bounded by
  $\Pr[ \dd( \hat \phi ,\hat \phi') )\leq 5/2^n] \geq  2/3$. This is true for each term in Equation~\eqref{eq:PsiA}, and therefore also for $P(\hat \phi, \hat \phi'\vert \psi^A_{\phi \alpha\phi' \alpha'} )$.
  
Second we show that if 
$\vert \psi^A_{\phi \alpha\phi' \alpha'}\rangle
\in \HH^{\geq 9/2^{n}} $, that is
 if $\dd( \phi ,\phi') \geq  9/ 2^{-n}$, then
 $\Pr[ \dd( \hat \phi ,\hat \phi')  \leq  5/2^n] \leq  1/3$,
 i.e. the acceptance probability is less or equal than $1/3$.
  
To this end reason again for each term Equation~\eqref{eq:PsiA} separately. Use again the triangle inequality, and note that if $ \dd(\phi, \phi') \geq 9/2^n$ and 
  $\dd(\hat \phi, \hat \phi') \leq 5/2^n$, then either $ \dd(\hat \phi, \phi) \geq 2/2^n$ or $ \dd(\phi', \hat \phi') \geq 2/2^n$.
  Hence if $\dd( \phi ,\phi') \geq 9/ 2^{-n}$, then
 $\Pr[ \dd( \hat \phi ,\hat \phi')  \leq  5/2^n] \leq  1/3$.

{\bf The set of neighbouring states $\HH^{\leq 2^{-n}} $ is non-empty.}
Since $U_n$ acts on $n$ qubits, it has $2^n$ eigenstates, which form an orthonormal basis of the Hilbert space with eigenphases in $\phi\in [0,1)$. 
By the pigeonhole principle, there must be at least 2 eigenstates with eigenphases $\phi$, $\phi'$ satisfying $\dd( \phi ,\phi') \leq 2^{-n}$. 
\end{proof}

The following problem is a generalisation of the problem 
\emph{ multiple copies of eigenstates of commuting $k$-local Hamiltonian}, see Definition \ref{Prob:MCopies} 
to the case where the input is an oracle implementing unitary transformations, rather than by a commuting $k$-local Hamiltonian.

\begin{deff} \label{Prob:MultCopyU}
{\bf Multiple copies of eigenstates $U$.}
Given access to oracle $C$, denote by
$T^{eq}(n)$ the set of triples of equal eigenvalues, and by $\HH^{eq}(n)$ the corresponding subspace:
\begin{eqnarray}
T^{eq}(n)&=& \{ ( (\phi_1,\alpha_1), (\phi_2,\alpha_2),(\phi_3,\alpha_3) )\in S^{\times  3} \nonumber\\
& & : \phi_1=\phi_2=\phi_3\} \nonumber\\
\HH^{eq} (n)&=& \SPAN (
\{ \vert \psi_{\phi_1 \alpha_1}\rangle \vert \psi_{\phi_2 \alpha_2}\rangle \vert \psi_{\phi_3 \alpha_3}\rangle 
\nonumber\\
& &: ( (\phi_1,\alpha_1), (\phi_2,\alpha_2),(\phi_3,\alpha_3) \in T^{eq}
\} )\ ;\nonumber\\
\end{eqnarray}
and denote by $T^{neq}(n)$ the set of triples of eigenvalues where 
at least two 
of are significantly different, and by 
$\HH^{neq}(n)$ the corresponding subspace:
\begin{eqnarray}
T^{neq}(n)&=& \{ ( (\phi_1,\alpha_1), (\phi_2,\alpha_2),(\phi_3,\alpha_3) \in S^{\times  3} 
\nonumber\\
& &: d(\phi_1,\phi_2)>14/2^n
\nonumber\\& & {\rm\quad OR\ } d(\phi_1,\phi_3)>14/2^n 
\nonumber\\& & {\rm\quad OR\ }
d(\phi_2,\phi_3)>14/2^n 
\} \nonumber\\
\HH^{neq} (n)&=& \SPAN (
\{ \vert \psi_{\phi_1 \alpha_1}\rangle \vert \psi_{\phi_2 \alpha_2}\rangle \vert \psi_{\phi_3 \alpha_3}\rangle 
\nonumber\\
& &: ( (\phi_1,\alpha_1), (\phi_2,\alpha_2),(\phi_3,\alpha_3) \in  T^{neq}
\} )\ .\nonumber\\
\end{eqnarray}
\end{deff}

\begin{thm}
Given access to oracle $C$, and given $n$, there exists a $2/3$-total quantum verification procedure $Q_{eq}$ acting on $3n$ qubits such that the corresponding pair of relations 
$(\HH^{\geq 2/3\, , \,C}_{Q_{eq}}(n, \vert \psi \rangle), \HH^{\leq 1/3\, , \,C}_{Q_{eq}}(n, \vert \psi \rangle))\in \TFQMA^C(2/3,1/3)$ satisfy
\begin{eqnarray}
\HH^{\geq 2/3\, , \,C}_{Q_{eq}}(n, \vert \psi \rangle)&\supseteq &
\HH^{eq}(n) \ ,\nonumber\\
\HH^{\leq 1/3\, , \,C}_{Q_{eq}}(n, \vert \psi \rangle))&\supseteq &
\HH^{neq}(n)  \ .\nonumber\\
&&
\end{eqnarray}
\end{thm}

\begin{proof}

{\bf Quantum verification procedure.} 
We first describe the quantum verification procedure $Q_{eq}$, which we view as acting on three $n$-qubit states.
Carry out the phase estimation algorithm on the three $n$ qubit states yielding outcomes $\hat \phi_1$, $\hat \phi_2$, $\hat \phi_3$. Accept if $\dd(\hat \phi_i , \hat \phi_j)\leq 10/2^n$ for the three pairs $(i,j)\in \{(1,2), (2,3), (3,1)\}$, otherwise reject.

{\bf Eigenbasis of the quantum verification procedure}
The basis of product states $
\{
\vert \psi_{\phi_1 \alpha_1}\rangle \vert \psi_{\phi_2 \alpha_2}\rangle \vert \psi_{\phi_3 \alpha_3}\rangle 
 \}$ constitute the eigenbasis of $Q_{eq}$. This follows from the remarks on the phase estimation algorithm made in the proof of Theorem \ref{Thm:AlmostDegU}.

{\bf Acceptance probability of states in $\HH^{eq}(n)$.}
We now show that on states of the form $\vert \psi_{\phi_1 \alpha_1}\rangle \vert \psi_{\phi_2 \alpha_2}\rangle \vert \psi_{\phi_3 \alpha_3}\rangle $ with $\phi_1=\phi_2=\phi_3$
the  quantum verification procedure $Q_{eq}$ will accept with probability greater than $2/3$. 

First, using Equation~\eqref{Eq:PhaseEst},  note that the probability that 
$\dd(\hat \phi_i , \phi_i)\leq k/2^n$ for $i=1,2,3$ simultaneously is lower bounded by 
$
\left(1-3/(2k-1)\right)$.

Second, the triangle inequality implies that if $\dd(\hat \phi_i , \phi_i)\leq k/2^n$ for $i=1,2,3$, then
 $\dd(\hat \phi_i , \hat \phi_j)\leq 2k/2^n$ for the three pairs $(i,j) \in \{(1,2), (2,3), (3,1)\}$.
 
 Setting $k=5$ yields the result.
 
 {\bf Acceptance probability of states in $\HH^{neq}$.}
The quantum verification algorithm accepts with probability less than $1/3$ on all states in $\HH^{neq}$.

Consider a state of the form
$\vert \psi_{\phi_1 \alpha_1} \rangle  \vert \psi_{\phi_2 \alpha_2} \rangle  \vert \psi_{\phi_3 \alpha_3} \rangle\in \HH^{neq}$. Consequently, there is at least one pair $(i,j)$ for which 
$\dd(\phi_i , \phi_j)>14/2^n$.
If this state accepts, then $\dd(\hat \phi_i , \hat \phi_j)\leq 10/2^n$.
Consequently, using the triangle inequality, either
 $\dd( \phi_i , \hat \phi_i) > 2/2^n$ or
$\dd( \phi_j , \hat \phi_j) > 2/2^n$. 
Using Equation~\eqref{Eq:PhaseEst} with $k=2$ shows that at least one of these events has probability less than $1/3$. Hence the overall acceptance probability of the state is less than $1/3$.
 
 {\bf $\HH^{eq}$ is non-empty.}
 Trivial.

\end{proof}

The problems based on Definitions \ref{Prob:AlmostDegU}  and \ref{Prob:MultCopyU}  
are expected to be hard because outputting an eigenstate of $U$ with a specified eigenvalue is expected to be hard in general. It is instructive however to consider variants of the problem that are easy. For instance outputting a random eigenstate of $U$ and the corresponding eigenvalue (up to precision $2^{-n}$) is easy: take the completely mixed state and run the phase estimation algorithm. The output of the algorithm will be an approximate eigenvalue $\hat \phi$, and the state after running the algorithm will be a superposition of eigenstates with eigenvalues close to $\hat \phi$. And if one carries out this procedure on one half of a maximally entangled state, one obtains a superposition of eigenstate times their complex conjugate (see remark in Section \ref{MultCopies}) This is why we request 3 copies in Definition \ref{Prob:MultCopyU}.

Note also that if we have additional information on the structure of $U$, constructing eigenstates may become easy. For instance suppose, as in Kitaev's factorization algorithm, that there is a set of orthogonal states on which $U$ acts like $U\vert \chi_j\rangle = \vert \chi_{j+1}\rangle$, where $j=0,...,N-1$, and where we identify $\vert \chi_{N}\rangle=\vert \chi_{0}\rangle$. Suppose 
also that we can efficiently implement the transformation $V$ which transforms the computational basis state $\vert j\rangle$ into $\vert \chi_j\rangle$: $V\vert j\rangle\vert 0\rangle
=\vert 0\rangle\vert \chi_j\rangle$. Then acting with $V$ on the state $N^{-1/2}\sum_{j=0}^N
e^{i2\pi jk/N}
\vert j\rangle\vert 0\rangle$ will yield an eigenstate of $U$ with eigenvalue $e^{i2\pi k/N}$.


 \section{Open Questions}\label{Sec:Disc}


We have provided several examples of problems belonging to $\TFQMA$, showing that it is an noteworthy complexity class. We sketch here some interesting open questions.

One of our examples is based on a quantum money scheme. Can one extend and define more precisely the relation between $\TFQMA$ and quantum money?

Can one find additional problems in $\TFQMA$?
Note that in the classical case there are many problems that belong to $\TFNP$, including some problems of real practical importance, such as local search problems and finding Nash equilibria. Are there problems of real practical importance in $\TFQMA$?

We have introduced some natural restrictions of 
$\QMA$ and $\TFQMA$:
gapped quantum verification procedures in  Section \ref{SubSec:GapQVP}, and 
$1$-- and/or $0$--quantum verification procedures in Section \ref{SubSec:10QVP}. 
Another natural restriction is to require that there is a unique witness, i.e. that the witness Hilbert space is one-dimensional. Can one find examples of this type?

When the witness is classical, the class $\QMA$ becomes $\CQMA$. When in addition the verifier is classical, one obtains the classical class $\MA$. One can define the functional problems associated to these  classes $\FCQMA$ and $\FMA$, and the corresponding total functions $\TFCQMA$ and $\TFMA$. In all these cases one can introduce gapped versions, and unique versions of the functional classes. Are there examples of problems that fall in these classes? (Note that \cite{AK07} provides evidence that the certificate for Group Non--Membership could be classical, in which case Theorem \ref{Thm:GNM} would have to be changed to reflect inclusion in  $\TFCQMA^B$).

What would be the consequences if some of these complexity classes coincide? What would be the consequences if some of these complexity classes are trivial, i.e. coincide with $\FBQP$, or with $\FBPP$ (the functional analog of $\BPP$)?

Total functional $\NP$ ($\TFNP$) can also be defined as the functional analog of $\NP\cap\coNP$ \cite{MP91}. We believe that one can similarly show that
$\TFQMA$ is the functional analog of $\QMA \cap \coQMA$. We hope to report on this result in a future publication.

In the case of $\TFNP$, there exist a number of syntactically defined subclasses which each contain some complete problems, such as Polynomial Local Search (PLS), Polynomial Parity Argument (PPA), Polynomial Parity Argument on a Directed Graph (PPAD), Polynomial Pigeonhole Principle (PPP). Are there  syntactically defined subclasses of $\TFQMA$? If these syntactically defined subclasses of $\TFQMA$ exist, do they have natural complete problems?
Do the syntactically defined subclasses of $\TFNP$ (such as PLS, PPA, PPAD, PPP, etc...) have quantum analogs?  
Could one show that the problems considered in section \ref{Sec:ProbTFQMA} are complete for some of these  syntactically defined subclasses. This would provide evidence for the hardness of these problems.
Note that several of the problems 
we have introduced are based on the pigeonhole principle which is  at the basis of class PPP. These problems may fit into a quantum analog of PPP.

\acknowledgments.

We thank Andr\'as Gily\'en, Han-Hsuan Lin, Frank Verstraete and Ronald de Wolf for useful discussions.
Our research was partially funded by the Singapore Ministry of Education and the National Research Foundation under grant
 R-710-000-012-135 and by the QuantERA ERA-NET Cofund project QuantAlgo. S.M. thanks the Center for Quantum Technologies, Singapore, where part of this work was carried out.

\bibliographystyle{plain}

\end{document}